\documentclass[draftclsnofoot,onecolumn]{IEEEtran}
\usepackage{amssymb,amsfonts,amsmath,amstext,mathrsfs}
\usepackage{latexsym}
\usepackage{epsfig,psfrag,color,graphics,epsf}
\usepackage{enumerate,cite}

\usepackage{amsmath}
\usepackage{amssymb}
\usepackage{amsfonts}
\usepackage{amsthm}
\usepackage{graphicx}
\usepackage{float}
\usepackage{graphics}
\usepackage{graphicx}
\usepackage{euscript}
\usepackage{psfrag}

\newtheorem{theorem}{Theorem}
\newtheorem{lemma}{Lemma}
\newtheorem{claim}{Claim}
\newtheorem{proposition}{Proposition}

\newtheorem{corollary}{Corollary}
\newtheorem{definition}{Definition}

\newcommand {\bq} {\mbox{\boldmath $q$}}

\newcommand {\bv} {\mbox{\boldmath $v$}}

\newcommand {\bx} {\mbox{\boldmath $x$}}
\newcommand {\by} {\mbox{\boldmath $y$}}

\newcommand {\bP} {\mbox{\boldmath $P$}}

\newcommand {\bV} {\mbox{\boldmath $V$}}
\newcommand {\bW} {\mbox{\boldmath $W$}}
\newcommand {\bX} {\mbox{\boldmath $X$}}
\newcommand {\bY} {\mbox{\boldmath $Y$}}

\newcommand{\calA}{{\cal A}}
\newcommand{\calB}{{\cal B}}
\newcommand{\calC}{{\cal C}}
\newcommand{\calD}{{\cal D}}
\newcommand{\calE}{{\cal E}}

\newcommand{\calG}{{\cal G}}

\newcommand{\calP}{{\cal P}}

\newcommand{\calW}{{\cal W}}
\newcommand{\calX}{{\cal X}}
\newcommand{\calY}{{\cal Y}}

\begin{document}

\sloppy

\title{A General Formula for the Mismatch Capacity
} 

\author{
  Anelia Somekh-Baruch\thanks{A.\ Somekh-Baruch is with the Faculty of Engineering at Bar-Ilan University, Ramat-Gan, Israel.  Email: somekha@biu.ac.il. This paper was submitted to the IEEE Transactions on Information Theory.}
}
\maketitle

\begin{abstract}

The fundamental limits of channels with mismatched decoding are addressed. A general formula is established for the mismatch capacity of a general channel, defined as a sequence of conditional distributions with a general decoding metrics sequence. We deduce an identity between the Verd\'{u}-Han general channel capacity formula, and the mismatch capacity formula applied to Maximum Likelihood decoding metric. Further, several upper bounds on the capacity are provided, and a simpler expression for a lower bound is derived for the case of a non-negative decoding metric. The general formula is specialized to the case of finite input and output alphabet channels with a type-dependent metric. 
The closely related problem of threshold mismatched decoding is also studied, and a general expression for the threshold mismatch capacity is obtained. As an example of threshold mismatch capacity, we state a general expression for the erasures-only capacity of the finite input and output alphabet channel. We observe that for every channel there exists a (matched) threshold decoder which is capacity achieving. Additionally, necessary and sufficient conditions are stated for a channel to have a strong converse. Csisz\'{a}r and Narayan's conjecture is proved for bounded metrics, providing a positive answer to the open problem introduced in \cite{CsiszarNarayan95}, i.e., that the "product-space" improvement of the lower random coding bound, $C_q^{(\infty)}(W)$, is indeed the mismatch capacity of the discrete memoryless channel $W$. We conclude by presenting an identity between the threshold capacity and $C_q^{(\infty)}(W)$ in the DMC case. 
\end{abstract}

\vspace{7cm}

\pagebreak

\section{Introduction}

Maximum likelihood (ML) decoding is the decoding rule which minimizes the average error probability in deciding among several hypotheses. In certain setups of channel coding, due to practical limitations such as errors in channel estimation or limited resources, the decoder has a fixed structure which does not match the actual channel over which information is transmitted. This setup is referred to as {\it mismatched decoding}. Mismatched decoding has been studied extensively, especially for discrete memoryless channels (DMCs). It is usually assumed that the decoding rule maximizes, among all the codewords, a certain accumulated metric between the channel output sequence and the codeword. The highest achievable rate using a given decoder is referred to as the {\it mismatch capacity} which is obtained by optimizing over the set of possible encoding strategies.

Achievable rates for the discrete memoryless mismatched channel using random coding were derived by  Csisz\'{a}r and K{\"o}rner \cite{CsiszarKorner81graph}
and by Hui \cite{Hui83}. Lapidoth \cite{Lapidoth96} introduced an improved lower bound on the mismatch capacity of the DMC by studying the achievable sum-rate of an appropriately chosen mismatched MAC, whose 
codebook is obtained by expurgating codewords from the product of the codebooks of the two users. In \cite{SomekhBaruchMismatch1Arxiv2013},\cite{SomekhBaruchISIT_2013} the achievable region and error exponents of a {\it cognitive} MAC were considered, using superposition coding or random binning, whose sum-rate serves as a lower bound on the capacity of the single user channel. An improved bound was presented by Scarlett et al.\ using a refinement of the superposition coding ensemble. For given auxiliary random variables, the results of \cite{SomekhBaruchMismatch1Arxiv2013,SomekhBaruchISIT_2013,ScarlettMartinezGuilleniFabregasISIT_2013} may yield improved achievable rates for the DMC.
In 
 \cite{CsiszarKorner81graph}, 
an error exponent for
random coding with fixed composition codes and mismatched decoding was established using a
graph decomposition theorem. 
For other related works and extensions see 
\cite{Balakirsky_conference_95,ShamaiKaplan1993information,MerhavKaplanLapidothShamai94,LiuHughes96,Lapidoth96b,GantiLapidothTelatar2000,ShamaiSason2002,ScarlettFabregas2012,ScarlettAlfonsoFabregas2013,ScarlettMartinezGuilleniFabregas2012AllertonSU,ScarlettPengMerhavMartinezGuilleniFabregasArxiv2013} and references therein.

Upper bounds on the mismatch capacity have received much less attention relatively to the lower bounds. Except for some special channels, the best known upper bound on the mismatch capacity is the capacity of the same channel with matched ML decoding. A converse theorem for the mismatched binary-input DMC was proved in \cite{Balakirsky95}, but in general, the problem of determining the mismatch capacity of the DMC or providing a non-trivial upper bound on it, remains open.

In \cite{CsiszarNarayan95}, the mismatch capacity of the DMC with decoding metric $q$, denoted $C_q(W)$, is considered. It is shown that the lower bound derived previously by Csisz\'{a}r and K{\"o}rner \cite{CsiszarKorner81graph}
and by Hui \cite{Hui83} is not tight in general but its positivity is a necessary condition for positive mismatch capacity. 
This result is established by proving that the random coding bound for the product channel $W_{Y_1,...,Y_K|X_1,...,X_K}=\prod_{i=1}^K W_{Y_i|X_i}$ ($K$ consecutive channel uses of the DMC $W$), denoted $C_q^{(K)}(W)$, may result in strictly higher achievable rates. They refer to the improved bound as the "product-space" improvement of the lower bound, and the supremum of the achievable rates obtained by taking the limit of $C_q^{(K)}(W)$ as $K$ tends to infinity is denoted $C_q^{(\infty)}(W)$. In the special case of erasures-only (e.o.) capacity, the product space improvement is shown to be tight, but the question of whether this bound is tight in general remains open, and it is conjectured to be tight. It is further stated in \cite{CsiszarNarayan95} that "although the bound is not computable, its tightness would afford some valuable conclusions, for instance, that for $R<C_q(W)$, codes with $d$-decoding always exist with rates approaching $R$ and probability of error approaching zero exponentially fast." Another implication of an affirmative answer to the conjecture concerns the threshold capacity of the DMC. The threshold capacity is the supremum of achievable rates obtained by decoding the unique message which accumulates a metric that exceeds a predetermined threshold.
It is stated in \cite{CsiszarNarayan95} that if the conjecture is true, then the threshold capacity and the mismatch capacity of the DMC are equal.

In this paper, the problem of mismatched decoding is addressed. A general formula for the mismatch capacity of a general channel, defined as a sequence of conditional distributions with a general decoding metrics sequence is established. 
We present two proofs for the upper bound on the mismatch capacity. 
The general capacity formula yields an identity between the Verd\'{u}-Han channel capacity formula, and the mismatch capacity formula applied to Maximum Likelihood decoding metric. 
Since the general capacity formula is not computable, we further provide two upper bounds on the capacity in terms of supremum over input processes of the infimum over a class of channels of the resulting spectral inf-mutual information rates. We also derive a simpler lower bound expression for the case of a non-negative decoding metric, including the special case of Mismatched Maximum Likelihood (MML) decoder, which is tailored for a channel which is different from the one over which transmission occurs.
Further, the general formula is specialized to the case of finite input and output alphabet channels with type-dependent metric.   
We study the closely related problem of threshold mismatched decoding, and obtain a general expression for the threshold mismatch capacity.
As an example of threshold mismatch capacity, we state a general expression for the erasures-only capacity of the finite input and output alphabet channel. 
We observe that for every channel there exists a (matched) threshold decoder which is capacity achieving.
We further provide necessary and sufficient conditions for a channel to have a strong converse.
Although the obtained expression of the general capacity formula is given in terms of a limiting expression which is not computable, 
it enables to prove Csisz\'{a}r and Narayan's conjecture, hence providing a positive answer to the Open Problem 6 introduced in \cite{CsiszarNarayan95}, i.e., that $C_q^{(\infty)}(W)$ is indeed the mismatch capacity of the DMC $W$.

The affirmative answer to Csisz\'{a}r and Narayan's conjecture results in an affirmative answer to the open problems 5 and 7 raised in  \cite{CsiszarNarayan95}, i.e., it can be concluded that: 
\begin{itemize}
\item There exist codes with rates approaching the mismatch capacity and probability of error decaying to zero exponentially fast as the block length goes to $\infty$. 
\item The threshold $d$-capacity of a DMC is equal to its mismatch $d$-capacity, at least when the metric $d$ is bounded.
\end{itemize}

The outline of this paper is as follows. 
Section \ref{sc: Notation} presents notation conventions. 
Section \ref{sc:  Preliminaries} provides a formal statement of the problem and definitions.  
In Section \ref{sc:  A General Formula for the Mismatch Capacity}, a general formula for the mismatch capacity is derived, a lower bound on the capacity for non-negative mismatched metric is derived, and two alternative upper bounds on the mismatch capacity are presented. 
The threshold capacity is addressed in Section \ref{sc: the threshold capacity}. 
The mismatch capacity of the DMC, as well as related special cases, are studied in section \ref{sc: The Mismatch Capacity of the DMC}.
In Section \ref{sc: The Random Coding Over a Given Codebook}, we analyze the random coding over a given codebook which result in an additional proof of the converse part of the coding theorem for the general formula of the mismatch capacity.
Section \ref{sc: Strong Converse} presents conditions for existence of a strong converse. 
Finally, Section \ref{sc: Conclusion} develops the concluding remarks.

\section{Notation}\label{sc: Notation}

Throughout this paper, scalar random variables are denoted by capital letters, their sample values are denoted by the respective lower case letters, and their alphabets are denoted by their respective calligraphic letters, e.g.\ $X$, $x$, and $\calX$, respectively. A similar convention applies to random vectors of dimension $n$ and their sample values, which are either denoted with the same symbols in the boldface font, e.g., $\bx=(x_1,...x_n)$ or superscripted by $n$, i.e., $x^n$. The set of all $n$-vectors with components taking values in a certain finite alphabet are denoted 
by the same alphabet superscripted by $n$, e.g., $\calX^n$. The notation $X\sim P$ will stand for $P$ being the distribution of the random variable $X$.

Information theoretic quantities such as entropy, conditional entropy, and mutual information are denoted following the usual conventions in the information theory literature, e.g., $H (X )$, $H (X |Y )$, $I(X;Y)$ and so on. To emphasize the dependence of the quantity on a certain underlying probability distribution, say $\mu$, it is subscripted by $\mu$, i.e., with notations such as $H_\mu(X )$, $H_\mu(X |Y)$, $I_\mu(X;Y)$, etc. The expectation operator is denoted by $\mathbb{E} \{\cdot\}$, and once again, to make the dependence on the underlying distribution $\mu$ clear, it is denoted by $\mathbb{E}_\mu \{\cdot\}$. The cardinality of a finite set $A$ is denoted by $|A|$. The indicator function of an event $\calE$ is denoted by $1\{\calE \}$. 

Let $\calP(\calX)$ denote the set of all probability distributions on $\calX$. 
For a given sequence $\by \in \calY^n$, $\calY$ being a finite alphabet,  $\hat{P}_{\by}$ denotes the empirical distribution on $\calY$ extracted from $\by$, in other words, $\hat{P}_{\by}$ is the vector $\{ \hat{P}_{\by} (y), y\in\calY\}$, where $ \hat{P}_{\by} (y)$ is the relative frequency of the symbol $y$ in the vector $\by$. The type-class of $\bx$ is the set of $\bx'\in\calX^n$ such that $\hat{P}_{\bx'}=\hat{P}_{\bx}$, which is denoted $T(\hat{P}_{\bx})$. 
The set of empirical distributions of order $n$ on alphabet $\calX$ is denoted  $\calP_n(\calX)$.

When $\calX$ is a finite alphabet we take a particular interest in the subset of $\calP(\calX^n)$, denoted $\calP_{CC}(\calX, n)$, which includes the p.m.f.'s which assign positive value to sequences that lie in a certain single type-class, i.e.,
\begin{flalign}\label{eq: P_CC dfn}
\calP_{CC}(\calX, n)= \left\{P\in \calP(\calX^n):\; \exists Q\in\calP_n(\calX) \mbox{ s.t. }P(x^n)=0\mbox{ }\forall x^n\notin T(Q) \right\},
\end{flalign}
note that $P$ need not necessarily be uniform within that type-class, nor does it necessarily assign positive value to all members of that type-class. 

For two sequences of positive numbers, $\{a_n\}$ and $\{b_n\}$, the notation $a_n\doteq b_n$ means that $\{a_n\}$ and $\{b_n\}$ are of the same exponential order, i.e., $\frac{1}{n} \ln \frac{a_n}{b_n}\rightarrow 0$ as $n\rightarrow \infty$. Similarly, $a_n\stackrel{\cdot}{\leq}b_n$ means that $\limsup_n \frac{1}{n} \ln \frac{a_n}{b_n} \leq 0$, and so on. 
Throughout this paper logarithms are taken to base $2$.

\section{Preliminaries} \label{sc:  Preliminaries}

In this paper, general single-user channels which are not restricted to be stationary memoryless nor ergodic are considered.
We adopt the following definition of \cite{VerduHan1994} of a general channel. 

\begin{definition}
A channel $\bW=W^{(n)},n=1,2,...$ is an arbitrary sequence of increasing dimension where $W^{(n)}$ is a conditional output
distribution from $\calX^n$ to $\calY^n$, where $\calX$ and $\calY$ are the
input and output alphabets, respectively\footnote{In fact, as in \cite{VerduHan1994}, the discussion can be easily extended to input alphabets which are not necessarily Cartesian products of increasing order of the same alphabet $\calX$.}. 
\end{definition}
With a little abuse of terminology we shall refer to $\bW$ as well as to $W^{(n)}$ as channels, where the exact meaning will be clear from the context. 
We note that, unless stated otherwise, throughout the paper, as in \cite{VerduHan1994}, we assume for simplicity that the alphabets $\calX$ and $\calY$ are finite or countably infinite. If either alphabet is general, proper modifications should be made such as replacing summations with integrals etc.

A rate-$R$ block-code of length $n$ consists of $ 2^{nR}$  $n$-vectors $\bx(m)$, $m = 1, 2, . . . , 2^{nR}$, which represent $2^{nR}$ different messages, i.e., it is defined by the encoding function
\begin{flalign}
f_n:\; \{1,...,2^{nR}\} \rightarrow \calX^n.
\end{flalign}
It is assumed that all possible messages are a-priori equiprobable, i.e., $P (m) = 2^{-nR}$ for all 
$m$, and the random message is denoted by $S$.

A mismatched decoder for the channel is defined by a mapping 
\begin{flalign}\label{eq: qn mapping}
q_n:\;  \calX^n\times \calY^n\rightarrow  \mathbb{R},
\end{flalign}
where the decoder declares that message $i$ was transmitted iff 
\begin{flalign}\label{eq: decoder decision rule}
q_n(\bx(i),\by)>q_n(\bx(j),\by), \forall j\neq i,
\end{flalign}
and if no such $i$ exists, an error is declared. 

Following are several useful definitions.
\begin{definition}
A code $\calC_n$ with decoding metric $q_n$ is an $\left(n,M_n,\epsilon\right)$-code for the channel $W^{(n)}$ if 
it has $M_n$ codewords of length $n$ and the average probability of error incurred by the decoder $q_n$ applied to the codebook $\calC_n$ and the output of the channel $W^{(n)}$ is no larger than $
\epsilon$. 
\end{definition}
In certain cases, it will be useful to omit the average probability of error $\epsilon_n$ from the notation and to refer to a code which has $M_n$ codewords of length $n$ as an $(n,M_n)$-code, it will also be useful to define the average probability of error associated with a codebook, a channel and a metric:
\begin{definition}\label{eq: P_e W calC_n q_n dfn}
For a given codebook $\calC_n$, let $P_e(W^{(n)},\calC_n,q_n)$ designate the average probability of error incurred by the decoder $q_n$ employed on the output of the channel $W^{(n)}$. 
\end{definition}
We next define an $\epsilon$-achievable rate and the mismatch capacity. 
\begin{definition}
A rate $R>0$ is an $\epsilon$-achievable rate for the channel $\bW$ with decoding metrics sequence $\bq=\{q_n\}_{n\geq 1}$ if for every $\delta>0$, there exists a sequence of codes $\{\calC_n\}_{n\geq 1}$ such that for all $n$ sufficiently large, $\calC_n$ is an $\left(n,M_n,\epsilon\right)$ code for the channel $W^{(n)}$ and decoding metric $q_n$ with rate $\frac{\log(M_n)}{n}>R-\delta$.  
\end{definition}
\begin{definition}
The capacity of the channel $\bW=\{W^{(n)}\}_{n\geq 1}$ with decoding metrics sequence $\bq=\{q_n\}_{n\geq 1}$ (or, the mismatch $\bq$-capacity of the channel $\bW$), denoted $C_{\bq}(\bW)$, is the supremum of rates that are $\epsilon$-achievable for all $0<\epsilon<1$. \end{definition}

A closely related notion to that of mismatched $q_n$-decoder is the $(q_n,\tau_n)$-threshold decoder which decides that $i$ is the transmitted message iff
\begin{flalign}
q_n(\bx(i),\by)\geq \tau_n \label{eq: threshold decision rule 1}
\end{flalign}
and
\begin{flalign}
q_n(\bx(j),\by)<\tau_n,\;\forall j\neq i.\label{eq: threshold decision rule 2}
\end{flalign}
We distinguish between two setups of threshold decoding.
\begin{definition}\label{df: threshold capacity}
The threshold $\bq$-capacity of a channel $\bW$, denoted $C_{\bq}^{thresh}(\bW)$, is defined as the supremum of rates attainable by
codes with $(q_n,\tau_n)$-threshold decoders of the form (\ref{eq: threshold decision rule 1})-(\ref{eq: threshold decision rule 2}) and any threshold sequence $\tau_n, n\geq 1$. 
\end{definition}
\begin{definition}\label{df: constant threshold capacity}
The constant threshold $\bq$-capacity of a channel, denoted $C_{\bq}^{const,thresh}(\bW)$, is defined as the supremum of the rates attainable by
codes with $(q_n,\tau)$-threshold decoders for a constant $\tau$ (which does not depends on $n$). 
\end{definition}
Clearly, 
\begin{flalign}
C_{\bq}^{const,thresh}(\bW)\leq C_{\bq}^{thresh}(\bW)\leq C_{\bq}(\bW),
\end{flalign}
where the last equality follows\footnote{For a formal proof of this claim see Lemma \ref{lm: threshold is inferior}.} since a threshold decoder (\ref{eq: threshold decision rule 1})-(\ref{eq: threshold decision rule 2}) is more restrictive than the mismatched decoder (\ref{eq: decoder decision rule}).

\section{A General Formula for the Mismatch Capacity}\label{sc:  A General Formula for the Mismatch Capacity}

In this section, we derive a general formula for the mismatch capacity. 
The general formula holds for general sequences of decoding metrics $q_n$ (\ref{eq: qn mapping}).  

The following notation will be useful in what follows. Let $\bq$ be a given sequence $\{q_i\}_{i\geq 1}$ of decoding metrics. 
For $\mu$, a distribution of a random variable $\tilde{X}^n$ on $\calX^n$, a real number $c$,  and an $n$-vector $\by\in\calY^n$, define the following function
\begin{flalign}
\Phi_{q_n}(c,\mu,\by)=&\mathbb{E}_{\mu}\left(1\left\{q_n(\tilde{X}^n,\by)\geq c \right\}\right)\nonumber\\
=& \mu\left\{q_n(\tilde{X}^n,\by)\geq c \right\}\nonumber\\
=&\sum_{\tilde{\bx}\in\calX^n:q_n(\tilde{\bx},\by)\geq c}\mu(\tilde{\bx}).
\end{flalign}
Although general alphabets are not treated in this paper, we note that in the case in which $\mu$ is a distribution on a general alphabet $\calX^n$, one has
\begin{flalign}\label{eq: Lebesgue integral}
\Phi_{q_n}(c,\mu,\by)\triangleq \int_{\tilde{\bx}:q_n(\tilde{\bx},\by)\geq c}d\mu(\tilde{\bx}),
\end{flalign}
where the notation $\int$ refers to Lebesgue integral.

Another term that will be used throughout the paper is the limit inferior in probability of a sequence of random variables. 
\begin{definition}\label{dfn: liminf in prob}\cite{VerduHan1994} 
The limit inferior in probability of a sequence of random variables $X_n, n\geq 1$, denoted $p\mbox{-limsup } X_n$, is the supremum of all $\alpha\in\mathbb{R}$ such that\footnote{Unlike the definition in \cite{VerduHan1994}, it is required that $\limsup_{n\rightarrow\infty}\mbox{Pr}\left\{X_n< \alpha\right\}=0$ rather than $\lim_{n\rightarrow\infty}\mbox{Pr}\left\{X_n< \alpha\right\}=0$ since in certain cases, the sequence $a_n=\mbox{Pr}\left\{X_n< \alpha\right\}$ might not converge to a limit. } $\limsup_{n\rightarrow\infty}\mbox{Pr}\left\{X_n< \alpha\right\}=0$, i.e., 
\begin{flalign}\label{eq: this is liminf dfn eq}
p\mbox{-liminf } X_n=\sup\{\alpha:\; \limsup_{n\rightarrow\infty}\mbox{Pr}\left\{X_n< \alpha\right\}=0\}
\end{flalign} 
\end{definition}
We note that we adopt the definition of \cite{Han2003} with strict inequality $\left\{X_n< \alpha\right\}$ rather than that of \cite{VerduHan1994} with loose inequality $\left\{X_n\leq \alpha\right\}$. The reason for this choice is explained in the sequel (see (\ref{eq: VerduHan UB iuglguygujyguyguyg})-(\ref{eq: equals 1})). 
\begin{definition}
The limit superior in probability of a sequence of random variables $X_n, n\geq 1$, denoted $p\mbox{-limsup } X_n$, is the infimum of all $\beta\in\mathbb{R}$ such that $\limsup_{n\rightarrow\infty}\mbox{Pr}\left\{X_n> \alpha\right\}=0$, i.e., 
\begin{flalign}\label{eq: this is limsup dfn eq}
p\mbox{-limsup } X_n=\inf\{\beta:\; \limsup_{n\rightarrow\infty}\mbox{Pr}\left\{X_n> \beta\right\}=0\}.
\end{flalign} 
\end{definition}
Define the set of sequences of distributions of increasing dimension
\begin{flalign}
\calP^{(\infty)}\triangleq \left\{\bP=\{P^{(n)}\}_{n\geq 1}:\; \forall n, P^{(n)}\in\calP(\calX^n) \right\}.
\end{flalign}
Define the subset of $\calP^{(\infty)}$ containing sequences of distributions which are uniform over their support, i.e., 
\begin{flalign}
\calP_{U}^{(\infty)}\triangleq \left\{\bP=\{P^{(n)}\}_{n\geq 1}\in  \calP^{(\infty)}: \forall n,\;P^{(n)}(\tilde{x}^n) = P^{(n)}(x^n) \mbox{ if }P^{(n)}(x^n)\cdot P^{(n)}(\tilde{x}^n)>0 \right\}
\end{flalign}

For a sequence of distributions $\bP=\{P^{(n)}\}_{n\geq 1}\in \calP^{(\infty)}$, a channel $\bW=\{W^{(i)}\}_{i\geq 1}$, and a sequence of metrics $\bq=\{q_i\}_{i\geq 1}$, let $(X^n,Y^n)\sim P^{(n)}\times W^{(n)}$ and denote 
\begin{flalign}
\underline{K}_{\bq}(\bP,\bW)\triangleq& p\mbox{-}\liminf  -\frac{1}{n}\log\left(\Phi_{q_n}\left(q_n(X^n,Y^n), P^{(n)},Y^n \right)\right).
\end{flalign}

The multi-letter expression for the mismatch capacity is stated in the following theorem. 
\begin{theorem}\label{th: General formula expression}
The mismatch $\bq$-capacity of the channel $\bW$ is given by
\begin{flalign}\label{eq: the General formula ksdjhfkjhsk}
C_{\bq}(\bW)= & \sup_{\bP\in \calP^{(\infty)}} \underline{K}_{\bq}(\bP,\bW) ,
\end{flalign}
where the supremum can be restricted to $\bP\in \calP_{U}^{(\infty)}$. 
\end{theorem}
Before we prove the theorem, a few comments are in order. 
\begin{itemize}
\item Observe that 
if $X^n$ and $Y^n$ designate the channel input and output, respectively, and $P^{(n)}$ is uniform over an $(n,M_n)$ codebook, $\Phi_{q_n}(q_n(X^n,Y^n) ,P^{(n)},Y^n)$ 
can be regarded as the conditional error probability given $(X^n,Y^n)$ in a single drawing of another codeword $\tilde{X}^n$ uniformly over the codebook.
Hence, the capacity formula is the supremum over input distribution sequence of the limit inferior in probability of the exponent of the conditional error probability in a single drawing of another codeword $\tilde{X}^n$ uniformly over the codebook.
More generally, for $\mu\in\calP(\calX^n)$, and a triple of random variables 
\begin{flalign}
(X^n,\tilde{X}^n,Y^n)\sim \mu(X^n)\mu(\tilde{X}^n)W^{(n)}(Y^n|X^n)
\end{flalign} we have
\begin{flalign}
\Phi_{q_n}\left(q_n(X^n,Y^n), P^{(n)},Y^n\right)=\mbox{Pr}\left\{ q_n(\tilde{X}^n,Y^n)\geq q_n(X^n,Y^n)|X^n,Y^n\right\}.
\end{flalign}
\item We note that a general metric can account for any decoder with disjoint decision regions. To realize this, note that a general decoder with disjoint decision regions $\calD_m, m=1,..,M_n$ applied to the codebook $\calC=\{\bx_m\}_{m=1}^{M_n}$ can be expressed as a decoder with respect to the metric
\begin{flalign}
q_n(\bx_m,\by)&= 1\{\by\in\calD_m\},\forall \bx_m\in\calC.
\end{flalign}
\item We also note that the proof of Theorem \ref{th: General formula expression} can be extended quite straightforwardly to rather general alphabets $\calX,\calY$ and appropriate $\sigma$-algebras, as long as for any probability distribution $P_{X^n}\in\calP(\calX^n)$, the probability distribution $P_{X^n}\times W_{Y^n|X^n}$ is well defined. In the general case, the definition of $\Phi$ should be replaced with the Lebesgue integral (\ref{eq: Lebesgue integral}). 
\item 
\end{itemize}
Next, Theorem \ref{th: General formula expression} is proved. As mentioned before, an additional proof of the converse part of the theorem is provided in Section \ref{sc: The Random Coding Over a Given Codebook}.
\begin{proof}
We begin with the proof of the converse part. 
The following lemma implies that the highest achievable rate is upper bounded by $\sup_{\bP} \underline{K}_{\bq}(\bP,\bW) $.
\begin{lemma}\label{lm: VerduHan Lemma}
Let $X^n$ be the random variable uniformly distributed over an $(n,M_n)$-code $\calC_n$, 
and $Y^n$ the output of a channel $W^{(n)}$ with $X^n$ as the input, then 
\begin{flalign}\label{eq: VerduHan UB}
\mbox{Pr}\left\{-\frac{1}{n} \log \left(\Phi_{q_n}(q_n(X^n,Y^n),P^{(n)},Y^n)\right) < \frac{1}{n}\log M_n \right\}= P_e(W^{(n)},\calC_n,q_n),
\end{flalign}
where $P^{(n)}$ is the distribution of the codeword $X^n$, i.e., a uniform distribution over $\calC_n$.
\end{lemma}
\begin{proof}
Note that
\begin{flalign}\label{eq: this equation explains strong converse}
&\Phi_{q_n}(q_n(X^n,Y^n),P^{(n)},Y^n) \nonumber\\
=&\sum_{\bx'\in\calC_n:\; q_n(\bx',Y^n)\geq q_n(X^n,Y^n)} P^{(n)}(\bx') \nonumber\\
=&\frac{|\{\bx'\in\calC_n:\; q_n(\bx',Y^n)\geq q_n(X^n,Y^n)\}|}{M_n} \nonumber\\
\end{flalign}
where the last equality follows since $X^n$ is distributed uniformly over the codebook of size $M_n$. Hence, the left hand side of (\ref{eq: VerduHan UB}) is equal to
\begin{flalign}
&\mbox{Pr}\left\{-\frac{1}{n} \log \left(|\{\bx'\in\calC_n :\; q_n(\bx',Y^n)\geq q_n(X^n,Y^n)\}|\right) < 0 \right\}\nonumber\\
=& \mbox{Pr}\left\{|\{\bx'\in\calC_n :\; q_n(\bx',Y^n)\geq q_n(X^n,Y^n)\}| > 1 \right\}\nonumber\\
=& P_e(W^{(n)},\calC_n,q_n),
\end{flalign}
where the last step follows since the decision rule (\ref{eq: decoder decision rule}) can be rewritten as: decide $m$ iff 
\begin{flalign}
  |\{\bx'\in\calC_n:\; q_n(\bx',\by)\geq q_n(\bx_m,\by)\}| =1
\end{flalign}
and the corresponding decision region $\calD_m$ is \begin{flalign}
\calD_m=& \left\{ \by:\;  |\{\bx'\in\calC_n:\; q_n(\bx',\by)\geq q_n(\bx_m,\by)\}| =1 \right\}.
\end{flalign}
This concludes the proof of Lemma \ref{lm: VerduHan Lemma}. 
\end{proof} 
Now, fix $\gamma>0$ and assume in negation that $R=\sup_{\bP'} \underline{K}_{\bq}(\bP',\bW) +2\gamma$ is an achievable rate, therefore, there exists a sequence of $(n,M_n)$-codes, $\{\calC_n\}_{n\geq 1}$, satisfying
$\limsup_{n\rightarrow\infty} P_e(W^{(n)},\calC_n,q_n)=0$ and $\liminf_{n\rightarrow\infty}\frac{1}{n}\log M_n \geq R>\sup_{\bP} \underline{K}_{\bq}(\bP,\bW) +\gamma$. Thus, from Lemma 
\ref{lm: VerduHan Lemma} we have for sufficiently large $n$,
\begin{flalign}
&P_e(W^{(n)},\calC_n,q_n)\nonumber\\
=&\mbox{Pr}\left\{-\frac{1}{n} \log \left(\Phi_{q_n}(q_n(X^n,Y^n),P^{(n)},Y^n)\right) < \frac{1}{n}\log M_n \right\}\nonumber\\
\geq & \mbox{Pr}\left\{-\frac{1}{n} \log \left(\Phi_{q_n}(q_n(X^n,Y^n),P^{(n)},Y^n)\right) < \sup_{\bP'} \underline{K}_{\bq}(\bP',\bW) +\gamma \right\}\label{eq: this is the supremum}\\
\geq & \mbox{Pr}\left\{-\frac{1}{n} \log \left(\Phi_{q_n}(q_n(X^n,Y^n),P^{(n)},Y^n)\right) < \underline{K}_{\bq}(\bP,\bW) +\gamma \right\},\label{eq: VerduHan UB second appearance}
\end{flalign}
and by definition of $ \underline{K}_{\bq}(\bP,\bW)$, the r.h.s.\ of (\ref{eq: VerduHan UB second appearance}) is bounded away from zero for infinitely many $n$'s, and hence $P_e(W^{(n)},\calC_n,q_n)$ cannot vanish in contradiction to the assumption.
We observe that since $P^{(n)}$ is uniform over $\calC_n$, the supremum in (\ref{eq: this is the supremum}) can be restricted to include only sequences of distributions that are uniform over a subset of their support, and this concludes the proof of the converse part of Theorem \ref{th: General formula expression}. 

Next, the direct part of Theorem \ref{th: General formula expression} is proved.

Let $\bP=\{P^{(n)}\}_{n\geq 1}$ be an arbitrary sequence of distributions where $P^{(n)}\in\calP(\calX^n)$. 
We use random coding with $P^{(n)}$ to generate $M_n=2^{nR}$ independent codewords constituting the codebook where 
\begin{flalign}
R=& \underline{K}_{\bq}(\bP,\bW)-2\gamma,
\end{flalign} 
and $\gamma>0$ can be chosen arbitrarily small. 
Denote by $\calA_n$ the set of pairs of $n$-vectors $(\bx,\by)\in\calX^n\times\calY^n$ such that $2^{n(\underline{K}_{\bq}(\bP,\bW)-\gamma)}\cdot \Phi_{q_n}(q_n(\bx,\by),P^{(n)},\by)\leq 1$. 
Note that by definition of $\underline{K}_{\bq}(\bP,\bW)$ we have $\mbox{Pr} \left\{\calA_n^c\right\}\rightarrow 0$ as $n$ tends to infinity. 
Further, the ensemble average probability of error can be computed as the probability of at least one "failure" in $M_n-1$ independent Bernoulli experiments, i.e., the average probability of error denoted $\bar{P_e}$ satisfies 
\begin{flalign}
\bar{P_e}=& \mathbb{E}\left[1-\left(1-\Phi_{q_n}(q_n(X^n,Y^n),P^{(n)},Y^n)  \right)^{M_n-1}\right] \nonumber\\
\stackrel{(a)}{\leq} & \mathbb{E} \min\left\{1,M_n\Phi_{q_n}(q_n(X^n,Y^n),P^{(n)},Y^n)\right\} \nonumber\\
\leq & \mathbb{E} \left( 1\{ (X^n,Y^n)\in \calA_n \} \min\left\{1,M_n\Phi_{q_n}(q_n(X^n,Y^n),P^{(n)},Y^n)\right\} \right)\nonumber\\
& +\mbox{Pr}\left\{\calA_n^c\right\}\nonumber\\
\leq  &2^{-n\gamma} \mathbb{E} \left(1\{ (X^n,Y^n)\in \calA_n \}  2^{n(\underline{K}_{\bq}(\bP,\bW)-\gamma)}\Phi_{q_n}(q_n(X^n,Y^n),P^{(n)},Y^n) \right)\nonumber\\
&+\mbox{Pr}\left\{\calA_n^c\right\}\nonumber\\
\stackrel{(b)}{\leq} &2^{-n\gamma}+\mbox{Pr}\left\{\calA_n^c\right\}
\end{flalign}
where $(X^n,Y^n)\sim P^{(n)}\times W^{(n)}$, $(a)$ follows from the union bound, and $(b)$ follows by definition of $\calA_n$. Therefore, we have $\bar{P_e}\rightarrow 0$.
The vanishing ensemble average probability of error ensures that there exists a sequence of deterministic codebooks of rate $R=\underline{K}_{\bq}(\bP,\bW)-2\gamma$ whose average probability of error using the decoding metric sequence $\bq$ vanishes. 
The capacity formula follows since $\bP$ is arbitrary, $\gamma$ can be made arbitrarily small, and by definition of the capacity as the supremum of all achievable rates.

\end{proof}

We note the following straightforward upper bound on $C_{\bq}(\bW)= \sup_{\bP} \underline{K}_{\bq}(\bP,\bW)$.
\begin{corollary}\label{cr: Upper Bound Expectation}
The mismatch $\bq$-capacity of the channel $\bW$ is upper bounded as follows
 \begin{flalign}\label{eq: the General formula ksdjhfkjhsk 22}
C_{\bq}(\bW) \leq & \sup_{\bP\in \calP_{U}^{(\infty)}}    \liminf_{n\rightarrow\infty} \frac{1}{n} \mathbb{E}\log\frac{1}{\Phi_{q_n}\left(q_n(X^n,Y^n), P^{(n)},Y^n \right)}.
\end{flalign}
\end{corollary} 
\begin{proof}
Let $\{A_n\}_{n\geq 1}$ be  a sequence of non-negative random variables, and denote $\underline{A}\triangleq p\mbox{-}\liminf A_n$ one has for all $\epsilon>0$, 
\begin{flalign}
\mathbb{E}(A_n)\geq E(A_n1\{A_n\geq \underline{A}-\epsilon \})\geq (\underline{A}-\epsilon) E(1\{A_n\geq\underline{A}-\epsilon \}).
\end{flalign}
Therefore, by definition of $\underline{A}$, for all sufficiently large $n$
\begin{flalign}
\mathbb{E}(A_n)\geq  (\underline{A}-\epsilon) (1-\epsilon),
\end{flalign}
and hence, 
\begin{flalign}\label{eq: inequality between liminf and expectation}
\liminf_{n\rightarrow\infty} \mathbb{E}(A_n)\geq  \underline{A}
\end{flalign}
The bound (\ref{eq: the General formula ksdjhfkjhsk 22}) follows by applying the inequality (\ref{eq: inequality between liminf and expectation}) to $A_n=-\frac{1}{n}\log\Phi_{q_n}(q_n(X^n,Y^n), P^{(n)},Y^n)$, which yields
 \begin{flalign}\label{eq: the inequality expectation without the sup}
 \underline{K}_{\bq}(\bP,\bW) \leq &   \liminf_{n\rightarrow\infty} \frac{1}{n} \mathbb{E}\log\frac{1}{\left(\Phi_{q_n}\left(q_n(X^n,Y^n), P^{(n)},Y^n \right)\right)},
\end{flalign}
and (\ref{eq: the General formula ksdjhfkjhsk 22}) follows by taking the supremum over $\bP\in\calP_U^{(\infty)}$. 
\end{proof}

We next present a lower bound on $C_{\bq}(\bW)$ for non-negative metrics in the spirit of \cite{fischer1978some}, \cite{ShamaiKaplan1993information}.
A non-negative metric $q_n$ satisfies $q_n(x^n,y^n)\geq 0$ for all $x^n,y^n$. 

For a given channel $\bW$, a non-negative decoding metrics sequence $\bv=\{v_n\}_{n\geq 1}$, and input distributions sequence $\bP=\{P^{(i)}\}_{i\geq 1}$, denote 
\begin{flalign}\label{eq: Theta dfn}
\underline{\Theta}_{\bv}(\bP,\bW)\triangleq &p\mbox{-}\liminf  \frac{1}{n}\log\frac{v_n(X^n,Y^n)}{\tilde{P}^{(n)}(Y^n)},
\end{flalign}
where $(X^n,Y^n)\sim P^{(n)}\times W^{(n)}$, and 
$\tilde{P}^{(n)}$ is defined as 
\begin{flalign}\label{eq: tilde P dfn}
\tilde{P}^{(n)}(y^n)=\sum_{x^n\in\calX^n} P^{(n)}(x^n) v_n(x^n,y^n).
\end{flalign}

\begin{theorem}\label{eq: general expression with I is tight}
Let $\bv=\{v_i\}_{i\geq 1}$ be a non-negative metrics sequence. The mismatch $\bv$-capacity satisfies
\begin{flalign}\label{eq: equivalent rate lkfdnklajk is A tight a a}
C_{\bv}(\bW)\geq  \sup_{\bP\in\calP^{(\infty)}}  \underline{\Theta}_{\bv}(\bP,\bW) .
\end{flalign}
 \end{theorem}
\begin{proof}
Let $P^{(n)}$ be given. 
Note that
\begin{flalign}\label{eq: this equation explains strong converse FINAL}
&\Phi_{v_n}(v_n(X^n,Y^n),P^{(n)},Y^n) \nonumber\\
=&\sum_{\bx'\in\calX^n:\; v_n(\bx',Y^n)\geq v_n(X^n,Y^n)} P^{(n)}(\bx') \nonumber\\
\stackrel{(a)}{\leq}&\sum_{\bx'\in\calX^n:\; v_n(\bx',Y^n)\geq v_n(X^n,Y^n)} P^{(n)}(\bx')\frac{v_n(\bx',Y^n)}{v_n(X^n,Y^n)} \nonumber\\
=&\frac{1}{v_n(X^n,Y^n)}  \sum_{\bx'\in\calX^n:\; v_n(\bx',Y^n)\geq v_n(X^n,Y^n)} P^{(n)}(\bx')v_n(\bx',Y^n) \nonumber\\
\stackrel{(b)}{\leq}&\frac{1}{v_n(X^n,Y^n)}  \sum_{\bx'\in\calX^n} P^{(n)}(\bx')v_n(\bx',Y^n) \nonumber\\
=&\frac{\tilde{P}^{(n)}(Y^n)}{v_n(X^n,Y^n)} ,
\end{flalign}
where $(a)$ and $(B)$ follow since $v_n$ is a non-negative metric. 
Therefore,
\begin{flalign}\label{eq: Theta Identity}
-\frac{1}{n} \log \left(\Phi_{v_n}(v_n(X^n,Y^n),P^{(n)},Y^n)\right) &\geq \frac{1}{n} \log\frac{v_n(X^n,Y^n)}{\tilde{P}^{(n)}(Y^n)},
\end{flalign}
and thus
\begin{flalign}\label{eq: Since the left hand side of}
p\mbox{-}\liminf -\frac{1}{n} \log \left(\Phi_{v_n}(v_n(X^n,Y^n),P^{(n)},Y^n)\right) &\geq p\mbox{-}\liminf \frac{1}{n} \log\frac{v_n(X^n,Y^n)}{\tilde{P}^{(n)}(Y^n)}.
\end{flalign}
Since taking the supremum over $\bP$, the left hand side of (\ref{eq: Since the left hand side of}) becomes $C_{\bv}(\bW)$ by Theorem \ref{th: General formula expression}, (\ref{eq: equivalent rate lkfdnklajk is A tight a a}) follows.

\end{proof}

A few comments are in order:
\begin{itemize}
\item We note that the result of Theorem \ref{eq: general expression with I is tight} 
holds for the following important class of non-negative metrics.
\begin{definition}
We say that a non negative metric $v_n:(\calX^n,\calY^n)\rightarrow \mathbb{R}^{+}$ is a mismatched maximum likelihood (MML) metric if $v_n(x^n,y^n)=V^{(n)}(y^n|x^n), \forall (x^n,y^n)\in\calX^n\times \calY^n$ where $V^{(n)}$ is a conditional distribution from $\calX^n$ to $\calY^n$, i.e., an ML decoder with respect to the channel $V^{(n)}$ with equiprobable messages.
\end{definition}
The class of MML decoders is relevant especially for setups in which a suboptimal decoder is used due to incorrect knowledge of the channel rather than practical limitations on its structure. 

We also note that the bound of Theorem \ref{eq: general expression with I is tight} is tight in the matched case, and when $\frac{1}{n}\log\frac{V^{(n)}(Y^n|X^n)}{\tilde{P}^{(n)}(Y^n)} $ converges in probability to its (possibly time-varying) expectation, we obtain, 
\begin{flalign}
 p\mbox{-}\liminf \frac{1}{n}\log\frac{V^{(n)}(Y^n|X^n)}{\tilde{P}^{(n)}(Y^n)}\rightarrow & \frac{1}{n}\mathbb{E}\log\frac{V^{(n)}(Y^n|X^n)}{\tilde{P}^{(n)}(Y^n)}\nonumber\\
= & \frac{1}{n}\mathbb{E}\log\frac{\mu^{(n)}(X^n|Y^n)}{P^{(n)}(X^n)}\nonumber\\
 =&\frac{1}{n}\left(I(X^n;Y^n)- D\left(P^{(n)}(X^n|Y^n)\|\mu^{(n)}(X^n|Y^n) \right)\right).
\end{flalign}
where $\mu^{(n)}(X^n|Y^n)$ is the posterior probability induced by $\tilde{P}^{(n)}\times V^{(n)}$ and thus the divergence $D\left(P^{(n)}(X^n|Y^n)\|\mu^{(n)}(X^n|Y^n) \right)$ expresses an upper bound on the mismatch loss.

\item Theorem \ref{eq: general expression with I is tight}  can also be extended to include lower bounded metrics in the following manner.
\begin{corollary}\label{eq: general expression with I is tight corollary}
Let $\bv=\{v_i\}_{i\geq 1}$ be a lower bounded metrics sequence, i.e., $v_n(x^n,y^n)\geq -|B|> -\infty,\forall n,(x^n,y^n)\in\calX^n\times\calY^n$. 
The mismatch $\bv$-capacity satisfies
\begin{flalign}\label{eq: equivalent rate lkfdnklajk is A tight a}
C_{\bv}(\bW)\geq  \sup_{\bP}  p\mbox{-}\liminf \frac{1}{n}\log\frac{v_n(X^n,Y^n)+|B|}{\sum_{x^n\in\calX^n}P^{(n)}(x^n)[v_n(x^n,Y^n)+|B|]} .
\end{flalign}
 \end{corollary}
\begin{proof}
The proof follows by applying (\ref{eq: this equation explains strong converse FINAL}) to the non-negative metric $v_n(X^n,Y^n)+|B|$ which defines the same decision regions as those of $v_n(X^n,Y^n)$. 
\end{proof}

\item The matched case:
For a given sequence of input distributions $\bP=\{P^{(i)}\}_{i\geq 1}$ and a channel $\bW=\{W^{(i)}\}_{i\geq 1}$ recall the definitions of the inf-information rate and the sup-information rate \cite{VerduHan1994} as
\begin{flalign}
\underline{I}(\bP,\bW)=& p\mbox{-}\liminf \frac{1}{n}\log \frac{W^{(n)}(Y^n|X^n)}{P_{Y^n}(Y^n)}\nonumber\\
\overline{I}(\bP,\bW)=& p\mbox{-}\limsup\frac{1}{n}\log \frac{W^{(n)}(Y^n|X^n)}{P_{Y^n}(Y^n)},
\end{flalign}
respectively, where $(X^n,Y^n)\sim P^{(n)}\times W^{(n)}$ and consider the matched decoding metric 
\begin{flalign}\label{eq: ML decoding rule}
q_n(x^n,y^n)=
 W^{(n)}(y^n|x^n) ,
\end{flalign}
i.e., $\bq=\bW$. 
Note that
\begin{flalign}
\underline{I}(\bP,\bW)=&  \underline{\Theta}_{\bW}(\bP,\bW),
\end{flalign}
where $ \Theta_{\bq}(\bP,\bW)$ is defined in (\ref{eq: Theta dfn}) and $ \Theta_{\bW}(\bP,\bW)= \Theta_{\bq}(\bP,\bW)|_{\bq=\bW}$, i.e., $\bq$ matches $\bW$.

We emphasize the inequality relation between $\underline{\Theta}_{\bV}(\bP,\bW)$ and $\underline{K}_{\bV}(\bP,\bW)$, 
when $\bV$ is an MML metric, which is stated in the following lemma. Similarly to the definition of $\underline{K}_{\bq}(\bW)$, define
\begin{flalign}\label{eq: overline K dfn}
\overline{K}_{\bq}(\bW)= \sup_{\bP}  \mbox{p-}\limsup  -\frac{1}{n}\log\left( \Phi_{q_n}(q_n(X^n,Y^n),P^{(n)},Y^n)\right).
\end{flalign}
\begin{lemma}\label{cr: Identity stronger Lemma}
For every channel $\bW$, every MML decoding metrics sequence $\bV$, and every sequence of distributions $\bP$ 
\begin{flalign}\label{eq: jfbkjsbd}
-\frac{1}{n} \log \left(\Phi_{V^{(n)}}(V^{(n)}(Y^n|X^n),P^{(n)},Y^n)\right) &\geq \frac{1}{n} \log\frac{V^{(n)}(Y^n|X^n)}{\tilde{P}^{(n)}(Y^n)},
\end{flalign}
where $\tilde{P}^{(n)}$ is defined in (\ref{eq: tilde P dfn}) and consequently
\begin{flalign}
\underline{\Theta}_{\bV}(\bP,\bW) \leq  \underline{K}_{\bV}(\bP,\bW),\label{eq: identity equation}
\end{flalign}
and
\begin{flalign}
\overline{\Theta}_{\bV}(\bP,\bW) \leq  \overline{K}_{\bV}(\bP,\bW),\label{eq: identity equation second}
\end{flalign}
\end{lemma}
\begin{proof}
The inequality (\ref{eq: jfbkjsbd}) was derived in (\ref{eq: Theta Identity}). 
\end{proof}

We note the following identity which stems from Theorem \ref{th: General formula expression}.
\begin{corollary}\label{cr: Identity Corollary}
The following identity holds for every channel $\bW=\{W^{(n)}\}_{n\geq 1}$,
\begin{flalign}
 \sup_{\bP}  \underline{I}(\bP,\bW)=  \sup_{\bP}   \underline{K}_{\bW}(\bP,\bW).\label{eq: identity equation a}
\end{flalign}
\end{corollary}
\begin{proof}
The left hand side of (\ref{eq: identity equation a}) is the general formula of the channel capacity in the matched case introduced by Verd\'{u} and Han \cite{VerduHan1994}. From Theorem \ref{th: General formula expression}, it follows that the right hand side of (\ref{eq: identity equation a}) is equal to the capacity with matched decoding metric $\bW$. Since (\ref{eq: ML decoding rule}) is nothing but the optimal ML decoding metric, it achieves capacity in the matched case and implies the equality (\ref{eq: identity equation a}). 
It should be noted that the optimal ML decoder breaks ties arbitrarily and the right hand side of (\ref{eq: identity equation a}) assumes that ties are considered as errors, but it is easily verified that considering ties as error does not reduce the achievable rate. 
\end{proof}

\item 
We next state an important comment on the definition of the limit inferior in probability:
Consider the equality (\ref{eq: VerduHan UB}) which states that for an $(n,M_n)$-code $\calC_n$ we have
\begin{flalign}\label{eq: VerduHan UB iuglguygujyguyguyg}
\mbox{Pr}\left\{-\frac{1}{n} \log \left(\Phi_{q_n}(q_n(X^n,Y^n),P^{(n)},Y^n)\right) < \frac{1}{n}\log M_n \right\}=P_e(W^{(n)},\calC_n,q_n).
\end{flalign}
Note that it is easily verified that if one considers a loose inequality in (\ref{eq: VerduHan UB iuglguygujyguyguyg}), one has
\begin{flalign}\label{eq: equals 1}
 \mbox{Pr}\left\{-\frac{1}{n} \log \left(\Phi_{q_n}(q_n(X^n,Y^n),P^{(n)},Y^n)\right) \leq \frac{1}{n}\log M_n \right\}=1.
\end{flalign}
To realize this, recall (\ref{eq: this equation explains strong converse}), which yields that the right hand of (\ref{eq: VerduHan UB iuglguygujyguyguyg}) is equal to
\begin{flalign}
&\mbox{Pr}\left\{-\frac{1}{n} \log \left(|\{\bx'\in\calC_n :\; q_n(\bx',Y^n)\geq q_n(X^n,Y^n)\}|\right) \leq 0 \right\}\nonumber\\
=& \mbox{Pr}\left\{|\{\bx'\in\calC_n :\; q_n(\bx',Y^n)\geq q_n(X^n,Y^n)\}| \geq 1 \right\}\nonumber\\
=& 1\end{flalign}
where the last step follows since by setting $\bx'=X^n\in \calC_n$ we have $q_n(\bx',Y^n)= q_n(X^n,Y^n)$. 

Comparing (\ref{eq: VerduHan UB iuglguygujyguyguyg}) and (\ref{eq: equals 1}) one realizes that the subtlety of a strict inequality in definition \ref{dfn: liminf in prob} (see (\ref{eq: this is liminf dfn eq})) is important in establishing the proof of Theorem \ref{th: General formula expression}. 
\end{itemize}
Let a channel $W^{(n)}$ be given. The following theorem provides sufficient conditions for another channel, $\tilde{W}^{(n)}$, to have average probability of error essentially (up to a vanishing gap) no larger than that of $W^{(n)}$. 
It is a direct consequence of Lemma \ref{lm: VerduHan Lemma}, and it will be useful in deriving upper bounds in the spirit of the general formula of the channel capacity in the matched case \cite{VerduHan1994}, and also in the derivation of the mismatch capacity of the DMC.

\begin{theorem}\label{cr: conditions corollary}
Let $X^n$ be a random vector uniformly distributed over a codebook $\calC_n$, let $W^{(n)}$ and $\tilde{W}^{(n)}$ be two channels from $\calX^n$ to $\calY^n$, whose outputs when fed by $X^n$ are denoted $Y^n$ and $\tilde{Y}^n$, respectively.
If there exist sequences $\zeta_n\geq 0$, $\eta_n\geq 0$, and $\tau_n, n\geq 1$ 
such that
\begin{flalign}
& \mbox{Pr}\left\{  \tau_n< q_n(X^n,Y^n) \right\} \leq \zeta_n \label{eq: first condition for corollary2 a}\\
&  \mbox{Pr}\left\{ q_n(X^n,\tilde{Y}^n) < \tau_n \right\} \leq \eta_n \label{eq: first condition for corollary2 b}
\end{flalign}
and
\begin{flalign}
&P_{\tilde{Y}^n}=P_{Y^n}\label{eq: first condition for corollary}
\end{flalign}
then 
\begin{flalign}
P_e(\tilde{W}^{(n)},\calC_n,q_n)&\leq P_e(W^{(n)},\calC_n,q_n)+\zeta_n+\eta_n.
\end{flalign}
\end{theorem}
\begin{proof}
Let $P^{(n)}$ be uniform over $\calC_n$ and denote 
\begin{flalign}
E_n( q_n(X^n,Y^n),P^{(n)},Y^n)\triangleq -\frac{1}{n} \log \left(\Phi_{q_n}(q_n(X^n,Y^n),P^{(n)},Y^n)\right).\label{eq: E_n dfn}
\end{flalign}
 From Lemma \ref{lm: VerduHan Lemma} we know that 
\begin{flalign}
P_e(W^{(n)},\calC_n,q_n)&=\mbox{Pr}\left\{E_n( q_n(X^n,Y^n),P^{(n)},Y^n) < \frac{1}{n}\log M_n \right\}\nonumber\\
&\geq \mbox{Pr}\left\{E_n( q_n(X^n,Y^n),P^{(n)},Y^n) < \frac{1}{n}\log M_n,\tau_n\geq q_n(X^n,Y^n) \right\}\nonumber\\
&\stackrel{(a)}{\geq} \mbox{Pr}\left\{E_n(\tau_n,P^{(n)},Y^n) < \frac{1}{n}\log M_n,\tau_n\geq q_n(X^n,Y^n) \right\}\nonumber\\
&\stackrel{(b)}{\geq} \mbox{Pr}\left\{E_n(\tau_n,P^{(n)},Y^n) < \frac{1}{n}\log M_n\right\}-\zeta_n\nonumber\\
&\stackrel{(c)}{=}\mbox{Pr}\left\{E_n(\tau_n,P^{(n)},\tilde{Y}^n) < \frac{1}{n}\log M_n\right\}-\zeta_n\nonumber\\
&\geq \mbox{Pr}\left\{E_n(\tau_n,P^{(n)},\tilde{Y}^n) < \frac{1}{n}\log M_n,  q_n(X^n,\tilde{Y}^n)\geq \tau_n\right\}-\zeta_n\nonumber\\
&\stackrel{(d)}{\geq} \mbox{Pr}\left\{E_n(q_n(X^n,\tilde{Y}^n),P^{(n)},\tilde{Y}^n) < \frac{1}{n}\log M_n,  q_n(X^n,\tilde{Y}^n)\geq  \tau_n\right\}-\zeta_n\nonumber\\
&\stackrel{(e)}{\geq}\mbox{Pr}\left\{E_n(q_n(X^n,\tilde{Y}^n),P^{(n)},\tilde{Y}^n) < \frac{1}{n}\log M_n\right\}-\zeta_n-\eta_n\nonumber\\
&= P_e(\tilde{W}^{(n)},\calC_n,q_n)-\zeta_n-\eta_n,
\end{flalign}
where $(a)$ and $(d)$ follow since $\Phi( c,\mu,\by)$ is non increasing in $c$, $(b)$ follows from (\ref{eq: first condition for corollary2 a}), $(c)$ follows from (\ref{eq: first condition for corollary}) 
and $(e)$ follows from (\ref{eq: first condition for corollary2 b}).
\end{proof}
Before we state an upper bound on the mismatch $\bq$-capacity which stems from Theorem \ref{cr: conditions corollary} we present the following definition.
\begin{definition}
For a given sequence of input distributions $\bP=\{P^{(i)}\}_{i\geq 1}$ and a sequence of metrics $\bq$, let $\calW_{\bq}(\bP,\bW)$ be the set of channels $\tilde{\bW}=\tilde{W}^{(n)},n\geq 1$ such that 
\begin{flalign}
\forall n, P_{Y^n}=P_{\tilde{Y}^n}, 
\end{flalign}
and there exists a sequence $\tau_n,n\geq1$ such that 
\begin{flalign}
p\mbox{-}\liminf \left(q_n(X^n,\tilde{Y}^n)-\tau_n\right) \geq 0\geq p\mbox{-}\limsup \left(q_n(X^n,Y^n)-\tau_n\right),\label{eq: first condition for Theorem 2}
\end{flalign} 
where $(X^n,Y^n)\sim P^{(n)}\times W^{(n)} $ and $(X^n,\tilde{Y}^n)\sim P^{(n)}\times \tilde{W}^{(n)} $.
\end{definition}
The following theorem presents an upper bound on the mismatch $\bq$-capacity in terms of the supremum (over sequences of input distributions) of the infimum over channels 
of the mutual information density rates of $(X^n,\tilde{Y}^n)$ where $\tilde{Y}^n$ is the output process.
\begin{theorem}\label{eq: general expression with I}
The mismatch $\bq$-capacity of the channel $\bW$ is upper bounded as follows
\begin{flalign}\label{eq: equivalent rate lkfdnklajk}
C_{\bq}(\bW)\leq & \sup_{\bP\in \calP_{U}^{(\infty)}} \inf_{\tilde{\bW}\in \calW_{\bq}(\bP,\bW)} \underline{I}(\bP,\tilde{\bW}).
\end{flalign}
\end{theorem}
\begin{proof}

Let $\{\calC_n\}_{n\geq 1}$ be a given code sequence, let $M_n$ stand for the cardinality of $\calC_n$, and let $\bP=\{P^{(i)}\}_{i\geq 1}$ be the corresponding distributions of the input vectors, i.e., $P^{(n)}$ is uniform over $\calC_n$. 
From (\ref{eq: first condition for Theorem 2}), we have that if $\tilde{\bW}\in\calW_{\bq}(\bP,\bW)$ then there exist vanishing sequences $\zeta_n,\eta_n\geq 0$ such that the conditions (\ref{eq: first condition for corollary2 a})-(\ref{eq: first condition for corollary}) are met for some sequence $\tau_n$. 
Now, from Theorem \ref{cr: conditions corollary}, we obtain that for all $n$
\begin{flalign}\label{eq: using the corollary a}
P_e(W^{(n)},\calC_n,q_n)+\zeta_n+\eta_n\geq P_e(\tilde{W}^{(n)},\calC_n,q_n),
\end{flalign}
and consequently
\begin{flalign}\label{eq: using the corollary}
\limsup_{n\rightarrow\infty} P_e(W^{(n)},\calC_n,q_n)\geq \limsup_{n\rightarrow\infty} P_e(\tilde{W}^{(n)},\calC_n,q_n).
\end{flalign}
From \cite[Theorem 1]{VerduHan1994}, we have for the channel $\tilde{W}^{(n)}$
\begin{flalign}\label{eq: matched case lemma}
P_e(\tilde{W}^{(n)},\calC_n,q_n)\geq \mbox{Pr}\left\{ \frac{1}{n}\log\frac{\tilde{W}^{(n)}(\tilde{Y}^n|X^n)}{P_{\tilde{Y}^n}(\tilde{Y}^n)}\leq \frac{1}{n}\log M_n -\gamma\right\}-e^{-n\gamma},
\end{flalign}
where $\gamma>0$ is an arbitrary constant, and where in fact, this inequality holds also for the matched case with ML metric corresponding to $\tilde{W}^{(n)}$.

Hence, if 
\begin{flalign}\label{eq: long claim}
R> p\mbox{-} \liminf \frac{1}{n}\log\frac{\tilde{W}^{(n)}(\tilde{Y}^n|X^n)}{P_{\tilde{Y}^n}(\tilde{Y}^n)}, 
\end{flalign}
the average probability of error of the codebook sequence $\{\calC_n\}_{n\geq 1}$, $P_e(\tilde{W}^{(n)},\calC_n,q_n)$, does not vanish, and from (\ref{eq: using the corollary}) neither does $P_e(W^{(n)},\calC_n,q_n)$. The proof of this claim follows similarly to equations (3.11)-(3.14) in \cite{VerduHan1994}.

Since the claim (\ref{eq: long claim}) holds for every $\tilde{\bW} \in\calW(\bP,\bW)$, we have that if
\begin{flalign}
R> \inf_{\tilde{\bW}\in \calW_{\bq}(\bP,\bW)} p\mbox{-} \liminf \frac{1}{n}\log\frac{\tilde{W}^{(n)}(\tilde{Y}^n|X^n)}{P_{\tilde{Y}^n}(\tilde{Y}^n)}, 
\end{flalign}
the average probability of error of the codebook sequence $\{\calC_n\}_{n\geq 1}$ does not vanish. 

The above derivation holds for a given sequence of codebooks $\{\calC_n\}_{n\geq 1}$ and its corresponding $\bP=\{P^{(n)}\}_{n\geq 1}$. Taking the supremum over $\bP\in\calP_U^{(\infty)}$ we obtain that if
\begin{flalign}
R>\sup_{\bP\in \calP_U^{(\infty)}}  \inf_{\tilde{\bW}\in \calW_{\bq}(\bP,\bW)}p\mbox{-} \liminf \frac{1}{n}\log\frac{\tilde{W}^{(n)}(\tilde{Y}^n|X^n)}{P_{\tilde{Y}^n}(\tilde{Y}^n)}, 
\end{flalign}
there exists no sequence of codebooks with vanishingly low average probability of error. 

\end{proof}
We next deduce an alternative upper bound on the general channel mismatch capacity which requires the following definition. \begin{definition}
For a given sequence of input distributions $\bP$, let $\calW_{\bq}'(\bP,\bW)$ be the set of channels $\tilde{\bW}=\tilde{W}^{(n)},n\geq 1$ such that 
\begin{flalign}
\forall n, P_{Y^n,q_n(X^n,Y^n)}=P_{\tilde{Y^n},q_n(X^n,\tilde{Y}^n)},
\end{flalign} 
where $(X^n,Y^n)\sim P^{(n)}\times W^{(n)} $ and $(X^n,\tilde{Y}^n)\sim P^{(n)}\times \tilde{W}^{(n)} $.
\end{definition}
In other words, $\calW_{\bq}'(\bP,\bW)$ is the set of channels which induce the same joint law of the channel output and the metric between the channel input and output as that induced by $\bW$. 
\begin{theorem}\label{th: General formula expression with Is and joint Y Q}
The mismatch $\bq$-capacity of the channel $\bW$ is upper bounded as follows
\begin{flalign}
C_{\bq}(\bW)\leq &
 \sup_{\bP\in \calP_{U}^{(\infty)}}  \inf_{\tilde{\bW}\in \calW_{\bq}'(\bP,\bW) }  \underline{I}(\bP,\tilde{\bW}). 
\end{flalign}
\end{theorem}
\begin{proof}
The theorem follows by noting that for every codebook $\calC_n$, one has $P_e(W^{(n)}, \calC_n,q_n)=P_e(\tilde{W}^{(n)}, \calC_n,q_n)$ for every channel $\tilde{W}^{(n)}$ whose output $\tilde{Y}^n$ shares the same joint law with $q_n(X^n,\tilde{Y}^n)$ as that of $(Y^n,q_n(X^n,Y^n))$, as the 
 left hand side of (\ref{eq: VerduHan UB}) is identical for both channels $W^{(n)}$ and $\tilde{W}^{(n)}$. 
Consequently, if $\tilde{\bW}\in \calW_{\bq}'(\bP,\bW)$ one has
\begin{flalign} 
\limsup_{n\rightarrow\infty} P_e(W^{(n)}, \calC_n,q_n)=\limsup_{n\rightarrow\infty}P_e(\tilde{W}^{(n)}, \calC_n,q_n).
\end{flalign}
 The rest of the derivation follows along the line of the proof of Theorem \ref{eq: general expression with I} from the step proceeding (\ref{eq: using the corollary}) by substituting $\calW_{\bq}(\bP,\bW) $ with $\calW_{\bq}'(\bP,\bW) $. 
\end{proof}

We note that it is not evident which of the upper bounds (of Theorem \ref{eq: general expression with I} or \ref{th: General formula expression with Is and joint Y Q}) is tighter since neither of the sets $\calW_{\bq}(\bP,\bW) ,\calW_{\bq}'(\bP,\bW) $ is a subset of the other.

\section{The Threshold Capacity}\label{sc: the threshold capacity}
In this section, we study the threshold capacity. We distinguish between two setups which will be referred to as: threshold decoding and constant threshold decoding. 
Recall the definition of a threshold decoder (\ref{eq: threshold decision rule 1})-(\ref{eq: threshold decision rule 2}), and the proceeding Definitions \ref{df: threshold capacity}-\ref{df: constant threshold capacity} of the threshold $\bq$-capacity and the constant threshold $\bq$-capacity of a channel $\bW$.
The following definition extends $P_e(W^{(n)},\calC_n,q_n)$ to the case of a threshold decoder, 
\begin{definition}\label{eq: P_e W calC_n q_n tau_n dfn}
For a given codebook $\calC_n$, let $P_e(W^{(n)},\calC_n,(q_n,\tau_n))$ designate the average probability of error incurred by the $(q_n,\tau_n)$-threshold decoder employed on the output of the channel $W^{(n)}$.
\end{definition}

We first prove the following straightforward lemma concerning threshold decoders.
\begin{lemma}\label{lm: threshold is inferior}
For every channel $W^{(n)}$, $(n,M_n)$-codebook $\calC_n$, mismatched decoder $q_n$, and threshold level $\tau_n$ one has
\begin{flalign}
P_e(W^{(n)}, \calC_n,(q_n,\tau_n))\geq \frac{1}{2}P_e(W^{(n)},\calC_n,q_n).
\end{flalign}
\end{lemma}
\begin{proof}
Let $X^n$ and $Y^n$ denote the input and output of the channel $W^{(n)}$, respectively, and let $\calC_n=\left\{\bx_j\right\}_{j=1}^{M_n}$. Recall that $S$ denotes the transmitted message, which is uniformly distributed over $\{1,...,M_n\}$. The following inequalities hold:
\begin{flalign}
&P_e(W^{(n)}, \calC_n,(q_n,\tau_n))\nonumber\\
=& \mbox{Pr}\left( \left\{q_n(X^n,Y^n)<\tau_n\right\} \cup \left\{ \exists j\neq S:\; q_n(\bx_j,Y^n)\geq \tau_n, q_n(X^n,Y^n)\geq \tau_n \right\}\right)\nonumber\\
\stackrel{(a)}{\geq}& \frac{1}{2}\left[ \mbox{Pr}\left( \left\{q_n(X^n,Y^n)<\tau_n\right\} \right) +\mbox{Pr}\left( \exists j\neq S:\; q_n(\bx_j,Y^n)\geq \tau_n, q_n(X^n,Y^n)\geq \tau_n \right) \right]\nonumber\\
\geq& \frac{1}{2}\left[ \mbox{Pr}\left( \left\{q_n(X^n,Y^n)<\tau_n\right\} \right) +\mbox{Pr}\left(\exists j\neq S:\; q_n(\bx_j,Y^n)\geq q_n(X^n,Y^n), q_n(X^n,Y^n)\geq \tau_n \right) \right]\nonumber\\
\geq& \frac{1}{2}\left[ \mbox{Pr}\left( \left\{q_n(X^n,Y^n)<\tau_n\right\} \right) +\mbox{Pr}\left( \exists j\neq S:\; q_n(\bx_j,Y^n)\geq q_n(X^n,Y^n)\right)
+\mbox{Pr}\left(  q_n(X^n,Y^n)\geq \tau_n \right) -1\right]\nonumber\\
= &\frac{1}{2}\cdot\mbox{Pr}\left( \exists j\neq S:\; q_n(\bx_j,Y^n)\geq q_n(X^n,Y^n)\right)\nonumber\\
= &\frac{1}{2}P_e(W^{(n)},\calC_n,q_n),
\end{flalign}
where $(a)$ follows since for two events $A,B$, one has $ \mbox{Pr}\left( A\cup B\right)\geq \max \left\{\mbox{Pr}\left( A\right), \mbox{Pr}\left( B\right) \right\} \geq \frac{1}{2}\left(\mbox{Pr}\left( A\right)+\mbox{Pr}\left( B\right)\right)$.
\end{proof}
Consequently, $P_e(W^{(n)}, \calC_n,(q_n,\tau_n))$ cannot vanish as $n$ tends to infinity unless $P_e(W^{(n)},\calC_n,q_n)$ does, and the rate achieved by the threshold decoder $(q_n,\tau_n)$ cannot exceed the rate achieved by the decoder $q_n$.

 The following lemma is the equivalent of Lemma \ref{lm: VerduHan Lemma} to the threshold decoding case, it will serve to prove the converse result. 
\begin{lemma}\label{lm: VerduHan Lemma threshold}
Let $X^n$ be the random variable uniformly distributed over an $(n,M_n)$-code $\calC_n$, 
and $Y^n$ the output of a channel $W^{(n)}$ with $X^n$ as the input, then 
\begin{flalign}\label{eq: VerduHan UB threshold}
&P_e(W^{(n)},\calC_n,(q_n,\tau_n))\nonumber\\
=&\mbox{Pr}\left\{\left\{-\frac{1}{n} \log \left(\Phi_{q_n}(\tau_n,P^{(n)},Y^n)\right) < \frac{1}{n}\log M_n , \;q_n(X^n,Y^n)\geq \tau_n\right\}\cup \left\{q_n(X^n,Y^n)< \tau_n\right\}\right\} ,
\end{flalign}
where $P^{(n)}$ is the distribution of the codeword $X^n$, i.e., a uniform distribution over $\calC_n$.
\end{lemma}
The proof of the lemma follows similarly to that of Lemma \ref{lm: VerduHan Lemma} but it is included for the sake of completeness.
\begin{proof}
Note that
\begin{flalign}\label{eq: this equation explains strong converse threshold}
&\Phi_{q_n}(\tau_n,P^{(n)},Y^n) \nonumber\\
=&\sum_{\bx'\in\calC_n:\; q_n(\bx',Y^n)\geq \tau_n} P^{(n)}(\bx') \nonumber\\
=&\frac{|\{\bx'\in\calC_n:\; q_n(\bx',Y^n)\geq \tau_n\}|}{M_n} \nonumber\\
\end{flalign}
where the last equality follows since $X^n$ is distributed uniformly over the codebook of size $M_n$. Hence, the right hand side of (\ref{eq: VerduHan UB threshold}) is equal to
\begin{flalign}
&\mbox{Pr}\left\{\left\{-\frac{1}{n} \log \left(|\{\bx'\in\calC_n :\; q_n(\bx',Y^n)\geq \tau_n\}|\right) < 0, \;q_n(X^n,Y^n)\geq \tau_n\right\}\cup \left\{q_n(X^n,Y^n)< \tau_n \right\}\right\}\nonumber\\
=& \mbox{Pr}\left\{\left\{|\{\bx'\in\calC_n :\; q_n(\bx',Y^n)\geq \tau_n\}| > 1 , \;q_n(X^n,Y^n)\geq \tau_n\right\}\cup \left\{q_n(X^n,Y^n)< \tau_n\right\}\right\}\nonumber\\
=& P_e(W^{(n)},\calC_n,(q_n,\tau_n)),
\end{flalign}
where the last step follows since the decision rule (\ref{eq: threshold decision rule 1})-(\ref{eq: threshold decision rule 2}) can be rewritten as: decide $m$ iff 
\begin{flalign}
  |\{\bx'\in\calC_n:\; q_n(\bx',\by)\geq \tau_n\}| =1\mbox{ and } q_n(\bx_m,\by)\geq \tau_n
\end{flalign}
and the corresponding decision region $\calD_m$ is \begin{flalign}
\calD_m=& \left\{ \by:\;  |\{\bx'\in\calC_n:\; q_n(\bx',\by)\geq \tau_n\}| =1 \mbox{ and } \left\{q_n(\bx_m,\by)\geq \tau_n\right\}\right\}.
\end{flalign}
This concludes the proof of Lemma \ref{lm: VerduHan Lemma threshold}. 
\end{proof} 

Next, we obtain an asymptotic expression for $P_e(W^{(n)},\calC_n,(q_n,\tau))$ of a given sequence of codebooks and the corresponding $\bP$, with constant threshold level $\tau<\tau^*_{\bq}(\bP,\bW)$ where
\begin{flalign}\label{eq: liminf thresh}
\tau^*_{\bq}(\bP,\bW)= p\mbox{-}\liminf q_n(X^n,Y^n).
\end{flalign}
\begin{lemma}\label{lm: VerduHan UB threshold kfshfkjhbjkhbj}
Let $\{\calC_n\}_{n\geq 1}$ be a sequence of codes, where $\forall n, \calC_n$ is an  $(n,M_n)$-code and let $P^{(n)}$ be the uniform distribution over $\calC_n$. Denote $\bP=\{P^{(n)}\}_{n\geq 1}$, and let $\bW=\{W^{(n)}\}_{n\geq 1}$ be a given channel. For all $\tau<\tau^*_{\bq}(\bP,\bW)$ one has 
\begin{flalign}\label{eq: VerduHan UB threshold kfshfkjhbjkhbj}
\lim_{n\rightarrow\infty}\left|P_e(W^{(n)},\calC_n,(q_n,\tau))-\mbox{Pr}\left\{-\frac{1}{n} \log \left(\Phi_{q_n}(\tau,P^{(n)},Y^n)\right) < \frac{1}{n}\log M_n\right\}\right|=0,
\end{flalign}
where $(X^n,Y^n)\sim P^{(n)}\times W^{(n)}$.
\end{lemma}
\begin{proof}
By definition of the limit inferior in probability, since $\tau<\tau^*_{\bq}(\bP,\bW)$, there exists a vanishing sequence $\xi_n$ such that
\begin{flalign}
\mbox{Pr}\left\{q_n(X^n,Y^n)<\tau\right\}\leq \xi_n, \; \forall n.
\end{flalign}
Consequently, from Lemma \ref{lm: VerduHan Lemma threshold}, on one hand, 
\begin{flalign}
P_e(W^{(n)},\calC_n,(q_n,\tau))\leq&\mbox{Pr}\left\{-\frac{1}{n} \log \left(\Phi_{q_n}(\tau,P^{(n)},Y^n)\right) < \frac{1}{n}\log M_n , \;q_n(X^n,Y^n)\geq \tau\right\}
+\mbox{Pr}\left\{q_n(X^n,Y^n)< \tau\right\}\nonumber\\
\leq&\mbox{Pr}\left\{-\frac{1}{n} \log \left(\Phi_{q_n}(\tau,P^{(n)},Y^n)\right) < \frac{1}{n}\log M_n\right\}+\xi_n,
\end{flalign}
and on the other hand, from Lemma \ref{lm: VerduHan Lemma threshold}
\begin{flalign}
P_e(W^{(n)},\calC_n,(q_n,\tau))\geq&\mbox{Pr}\left\{-\frac{1}{n} \log \left(\Phi_{q_n}(\tau,P^{(n)},Y^n)\right) < \frac{1}{n}\log M_n , \;q_n(X^n,Y^n)\geq \tau\right\}\nonumber\\
\geq &\mbox{Pr}\left\{-\frac{1}{n} \log \left(\Phi_{q_n}\left( \tau,P^{(n)},Y^n\right) \right)< \frac{1}{n}\log M_n\right\}-\xi_n.
\end{flalign}
\end{proof}

From the above lemmas and discussion we deduce the expression for the general formula for the constant threshold capacity.
\begin{theorem}\label{th: General formula expression threshold}
The constant threshold capacity of the channel $\bW$ with decoding metrics sequence $\bq$ is given by
\begin{flalign}
\underline{C}_{\bq}^{const,thresh}(\bW)\triangleq & \sup_{\bP\in \calP^{(\infty)}} \sup_{\tau<\tau^*_{\bq}(\bP,\bW)} \mbox{p-}\liminf -\frac{1}{n}\log\left(\Phi_{q_n}\left(\tau, P^{(n)},Y^n \right)\right) ,
\end{flalign}
where $\bP$ is a sequence of distributions $P^{(n)}\in \calP(\calX^n), n\geq 1$, $(X^n,Y^n)\sim P^{(n)}\times W^{(n)} $, and $\tau^*_{\bq}(\bP,\bW)$ is defined in (\ref{eq: liminf thresh}). Further, the supremum can be restricted to $\bP\in \calP_{U}^{(\infty)}$. 
\end{theorem}
\begin{proof}
From (\ref{eq: VerduHan UB threshold}) it is clear that if one chooses a threshold level $\tau>\tau^*_{\bq}(\bP,\bW)$, one will have $P_e(W^{(n)},\calC_n,(q_n,\tau))$ bounded away from zero for infinitely many $n$'s, i.e., $\limsup_nP_e(W^{(n)},\calC_n,(q_n,\tau))>0$. We therefore compute the average probability of error using $\tau<\tau^*_{\bq}(\bP,\bW)$ arbitrarily close to $\tau^*_{\bq}(\bP,\bW)$. 
The proof of the converse part follows similarly to that of Theorem \ref{th: General formula expression} with Lemma \ref{lm: VerduHan UB threshold kfshfkjhbjkhbj} replacing Lemma 
\ref{lm: VerduHan Lemma}.

It remains to prove the direct part of Theorem \ref{th: General formula expression threshold}. The proof follows along the line of proof of the direct part of Theorem \ref{th: General formula expression} with minor modifications, but we repeat it here for the sake of completeness.

Let $\bP=\{P^{(n)}\}_{n\geq 1}$ be an arbitrary sequence of distributions where $P^{(n)}\in\calP(\calX^n)$. 
We employ random coding with $P^{(n)}$ to generate $M_n=2^{nR}$ independent codewords constituting the codebook where 
\begin{flalign}
R=& \left[\mbox{p-}\liminf  -\frac{1}{n}\log\left(\Phi_{q_n}(\tau,P^{(n)},Y^n)\right)\right]-2\gamma\nonumber\\
\triangleq & D_{\bq}(\bP,\bW,\tau)-2\gamma,
\end{flalign} 
and $\gamma>0$ can be chosen arbitrarily small. 
Denote by $\tilde{\calA}_n$ the set of pairs of $n$-vectors $(\bx,\by)\in\calX^n\times\calY^n$ such that $2^{n(D_{\bq}(\bP,\bW,\tau)-\gamma)}\cdot \Phi_{q_n}(\tau,P^{(n)},\by)\leq 1$. 
Note that by definition of $D_{\bq}(\bP,\bW,\tau)$ we have $\mbox{Pr} \left\{\tilde{\calA}_n^c\right\}\rightarrow 0$ as $n$ tends to infinity. 
Note also that since $\tau<\tau^*_{\bq}(\bP,\bW)$, there exists a sequence $\zeta_n$ converging to zero such that $\mbox{Pr}\left\{q_n(X^n,Y^n)<\tau\right\}\leq \zeta_n$. Further, the ensemble average probability of error can be computed as the probability of at least one "failure" in $M_n-1$ independent Bernoulli experiments, i.e., the average probability of error denoted $\bar{P}_e$ satisfies 
\begin{flalign}
\bar{P}_e\leq & \mathbb{E}\left[1-\left(1-\Phi_{q_n}(\tau,P^{(n)},Y^n)  \right)^{M_n-1}\right] +\mbox{Pr}\left\{q_n(X^n,Y^n)<\tau\right\}\nonumber\\
\stackrel{(a)}{\leq} & \mathbb{E} \min\left\{1,M_n\Phi(  \tau,P^{(n)},Y^n)\right\}+\zeta_n \nonumber\\
\leq & \mathbb{E} \left( 1\{ (X^n,Y^n)\in \tilde{\calA}_n \} \min\left\{1,M_n\Phi(  \tau,P^{(n)},Y^n)\right\} \right)\nonumber\\
& +\mbox{Pr}\left\{\tilde{\calA}_n^c\right\}+\zeta_n\nonumber\\
\leq  &2^{-n\gamma} \mathbb{E} \left(1\{ (X^n,Y^n)\in \tilde{\calA}_n \}  2^{n(D_{\bq}(\bP,\bW,\tau)-\gamma)}\Phi(  \tau,P^{(n)},Y^n) \right)\nonumber\\
&+\mbox{Pr}\left\{\tilde{\calA}_n^c\right\}+\zeta_n\nonumber\\
\leq &2^{-n\gamma}+\mbox{Pr}\left\{\tilde{\calA}_n^c\right\}+\zeta_n
\end{flalign}
where $(X^n,Y^n)\sim P^{(n)}\times W^{(n)}$, and $(a)$ follows from the union bound. Therefore, we have $\bar{P}_e\rightarrow 0$ as $n$ tends to infinity.
The vanishing ensemble average probability of error ensures that there exists a sequence of deterministic codebooks of rate $R=D_{\bq}(\bP,\bW,\tau)-2\gamma$, whose average probability of error vanishes. 
The capacity formula follows since $\bP$ is arbitrary, since $\gamma$ can be made arbitrarily small, and by definition of the capacity as the supremum of all achievable rates.

\end{proof}

We next determine the expression for the threshold capacity. 
\begin{theorem}\label{th: .ksjdf .jdfhkv.jh}
The threshold capacity of the channel $\bW$ with decoding metrics sequence $\bq$ is given by
\begin{flalign}\label{eq: .ksjdf .jdfhkv.jh}
\underline{C}_{ \bq}^{thresh}(\bW)\triangleq & \sup_{\bP\in \calP^{(\infty)}} \sup_{\tau_n, n\geq 1} \mbox{p-}\liminf -\frac{1}{n}\log\left(\Phi_{q_n}\left(\tau_n, P^{(n)},Y^n \right)\right) ,
\end{flalign}
where $\bP$ is a sequence of distributions $P^{(n)}\in \calP(\calX^n), n\geq 1$, $(X^n,Y^n)\sim P^{(n)}\times W^{(n)} $,  and the supremum is over sequences $\tau_n,n\geq 1$ that satisfy
\begin{flalign}\label{eq: condition on tau_n}
\limsup_{n\rightarrow\infty}\mbox{Pr}\left\{q_n(X^n,Y^n)<\tau_n\right\}=0.
\end{flalign}
Further, the supremum can be restricted to $\bP\in \calP_{U}^{(\infty)}$. 
\end{theorem}
\begin{proof}
The direct part follows along the lines of the proof of Theorem \ref{th: General formula expression threshold}, substituting $\tau$ by $\tau_n$, and by (\ref{eq: condition on tau_n}).
The converse part follows from (\ref{eq: VerduHan UB threshold}) and since (\ref{eq: condition on tau_n}) implies that there exists a vanishing sequence $\zeta_n$ such that  $\mbox{Pr}\left\{q_n(X^n,Y^n)<\tau_n\right\}\leq \zeta_n$, and consequently 
\begin{flalign}
&\mbox{Pr}\left\{\left\{-\frac{1}{n} \log \left(\Phi_{q_n}(\tau_n,P^{(n)},Y^n)\right) < \frac{1}{n}\log M_n , \;q_n(X^n,Y^n)\geq \tau_n\right\}\cup q_n(X^n,Y^n)< \tau_n\right\}\nonumber\\
\geq &\mbox{Pr}\left\{-\frac{1}{n} \log \left(\Phi_{q_n}(\tau_n,P^{(n)},Y^n)\right) < \frac{1}{n}\log M_n , \;q_n(X^n,Y^n)\geq \tau_n\right\}\nonumber\\
\geq &\mbox{Pr}\left\{-\frac{1}{n} \log \left(\Phi_{q_n}(\tau_n,P^{(n)},Y^n)\right) < \frac{1}{n}\log M_n \right\}-\zeta_n.
\end{flalign}
Additionally, the supremum over $\tau_n,n\geq 1$ in (\ref{eq: .ksjdf .jdfhkv.jh}) is constrained by (\ref{eq: condition on tau_n}) since (\ref{eq: VerduHan UB threshold})
implies that (\ref{eq: condition on tau_n}) is a necessary condition for vanishing average probability of error. To realize this, note that the left hand side of (\ref{eq: VerduHan UB threshold}) satisfies:
\begin{flalign}
&\mbox{Pr}\left\{\left\{-\frac{1}{n} \log \left(\Phi_{q_n}(\tau_n,P^{(n)},Y^n)\right) < \frac{1}{n}\log M_n , \;q_n(X^n,Y^n)\geq \tau_n\right\}\cup \left\{q_n(X^n,Y^n)< \tau_n\right\}\right\}\nonumber\\
\geq &\mbox{Pr}\left\{ q_n(X^n,Y^n)< \tau_n\right\}.
\end{flalign}
\end{proof}
We observe that for every channel there exists a (matched) threshold decoder which is capacity achieving.
\begin{proposition}
For every channel $\bW$ there exists a threshold decoder which achieves the matched capacity.
\end{proposition}
\begin{proof}
Consider the decoding rule (see  \cite{VerduHan1994}): decide $i$ iff $\bx_i$ is the unique vector satisfying
\begin{flalign}
\frac{1}{n}\log\frac{W^{(n)}(\by|\bx_i)}{P_{Y^n}(\by)}\geq \underline{I}(\bP,\bW)-\gamma. \label{eq: VH decoding rule}
\end{flalign}
This is a threshold decoder w.r.t.\ the metric $q_n(x^n,y^n)=\frac{1}{n}\log\frac{W^{(n)}(y^n|x^n)}{P_{Y^n}(y^n)}$. It was used by Verd\'{u} and Han \cite{VerduHan1994} to prove the direct part of the general channel capacity formula is capacity achieving for every channel.
\end{proof}
We conclude this section with an example. Consider the erasures-only capacity of the channel $\bW$, denoted $C_{eo}(\bW)$, which can be considered as the supremum of rates achievable by decoding with respect to the metric
\begin{flalign}
q_n^{eo}(x^n,y^n)=1\{W^{(n)}(y^n|x^n)>0\},\label{eq: eo metric for vectors}
\end{flalign}
with finite input and output alphabets. 
In this special case, it always holds that the actual transmitted codeword $X^n$ and received signal $Y^n$ satisfy
\begin{flalign}
q_n^{eo}(X^n,Y^n)=1.
\end{flalign}
Therefore, the erasures-only capacity $C_{eo}(\bW)$ is equal to the threshold capacity with fixed threshold level $\tau=1$. 
We obtain from Lemma \ref{lm: VerduHan Lemma threshold}, 
\begin{flalign}
P_e(W^{(n)},\calC_n,(q_n,\tau_n))|_{\tau_n=1}=\mbox{Pr}\left\{-\frac{1}{n}\log \left( \Phi_{q_n^{eo}}(1,P^{(n)},Y^n) \right)<\frac{1}{n}\log M_n\right\},
\end{flalign}
and consequently, the erasures-only capacity is stated in the following proposition.
\begin{proposition}\label{pr: erasures only proposition vectors}
The erasures-only capacity of the finite input and output alphabet channel $\bW$ is given by
\begin{flalign}
C_{eo}(\bW)=&\sup_{\bP\in \calP^{(\infty)}}p\mbox{-}\liminf -\frac{1}{n}\log \left( \Phi_{q_n^{eo}}(1,P^{(n)},Y^n) \right)\nonumber\\
=&\sup_{\bP\in \calP^{(\infty)}}p\mbox{-}\liminf -\frac{1}{n}\log \left(\sum_{\tilde{x}^n:\; W^{(n)}(Y^n|\tilde{x}^n)>0}P^{(n)}(\tilde{x}^n)\right),
\end{flalign}
where the supremum can be restricted to $\{\bP\in \calP_{U}^{(\infty)}\}$. 
\end{proposition}

\section{The Mismatch Capacity of the DMC and Related Special Cases}\label{sc: The Mismatch Capacity of the DMC}

We next address the important special case of a DMC with mismatched decoding. We focus on bounded additive metrics $q_n$. 
Consider a DMC with a finite input alphabet 
$\calX$ and finite output alphabet $\calY$, which is governed by the conditional p.m.f.\ $W$. 
As the channel is fed by an input vector $\bx \in\calX^n$, it generates an output vector $\by \in\calY^n$ according to the sequence of conditional probability distributions 
\begin{equation}P (y_i |x_1 , . . . , x_i , y_1 , . . . , y_{i-1} ) = W(y_i|x_i), \quad i = 1, 2, . . . , n\end{equation}
where for $i = 1, (y_1 , . . . , y_{i-1}) $ is understood as the null string. 

A special class of decoders is the class of additive decoding functions, i.e.,  
\begin{flalign}\label{eq: additive decoder decision rule}
 q_n(x^n,y^n)=\frac{1}{n}\sum_{i=1}^n q(x_i,y_i), 
\end{flalign}
where $q$ is a mapping from $\calX\times \calY$ to $\mathbb{R}$. It will be assumed that $q$ is bounded, that is
\begin{flalign}\label{eq: bounded q}
|q(x,y)|\leq B<\infty,\; \forall x\in\calX, y\in\calY.
\end{flalign}

As mentioned in the introduction, it was proved in \cite{CsiszarKorner81graph} and \cite{Hui83} that the mismatch $q^n$-capacity of the DMC, denoted $C_q(W)$, is lower bounded by
\begin{flalign}\label{eq: Cq1 dfn}
C_q^{(1)}(W)=& \max_{P_X} \underset{\tilde{P}\in\calW_1(P_X,W)}{\min}I_{\tilde{P}}(X;\tilde{Y}),
\end{flalign}
where for $P_X\in\calP(\calX)$, $\calW_1(P_X,W)$ is the following class of conditional probability distributions $\tilde{W}$ from $\calX$ to $\calY$
\begin{flalign}
\calW_1(P_X,W)=\left\{ \tilde{W}:  \tilde{P}_Y=P_Y, E_{\tilde{P}}(q(X,\tilde{Y}))\geq E_P(q(X,Y))  \right\},
\end{flalign}
$P=P_X\times W$, $\tilde{P}=P_X\times \tilde{W}$ and $(X,\tilde{Y})\sim P_X\times \tilde{W}$. 
As mentioned in \cite{CsiszarNarayan95}, by considering the achievable rate for the channel $W^n$ from $\calX^n$ to $\calY^n$, the following rate is also achievable
\begin{flalign}
C_q^{(n)}(W)=& \max_{P^{(n)}} \underset{\tilde{W}^{(n)}\in\calW_n(P^{(n)},W^n)}{\min}\frac{1}{n}I_{\tilde{P}}(X^n;\tilde{Y}^n),
\end{flalign}
where for $P^{(n)}\in\calP(\calX^n)$, $\calW_n(P^{(n)},W^n)$ is the following class of conditional probability distributions $\tilde{W}^{(n)}$ from $\calX^n$ to $\calY^n$
\begin{flalign}
\calW_n(P^{(n)},W^n)=\left\{ \tilde{W}^{(n)}:  \tilde{P}_{Y^n}=P_{Y^n}, E_{\tilde{P}}(q_n(X^n,\tilde{Y}^n))\geq E_P(q_n(X^n,Y^n))  \right\},
\end{flalign}
$P=P^{(n)}\times W^n$, $\tilde{P}=P^{(n)} \times \tilde{W}^{(n)}$, and $(X^n,\tilde{Y}^n)\sim P^{(n)}\times \tilde{W}^{(n)}$.

\begin{theorem}\label{eq: conjecture proof Theorem}
The mismatch capacity of the DMC $W$ 
with bounded additive decoding metric $q$ (\ref{eq: additive decoder decision rule})-(\ref{eq: bounded q}) is given by
\begin{flalign}\label{eq: dfn Cq infty W}
C_q^{(\infty)}(W) \triangleq \sup_{n\geq 1}C_q^{(n)}(W).
\end{flalign}
\end{theorem}

The proof of Theorem \ref{eq: conjecture proof Theorem} appears in Appendix \ref{ap: Proof of DMC Theorem}, it is divided into $9$ steps which are outlined as follows:
\begin{enumerate}
% 1
\item Without loss of asymptotic optimality in terms of achievable rates, one can assume that the codebook contains codewords that  lie in a single type-class. 
% 2
\item Applying an $(n,M_n,\epsilon_n)$-code
$\calC_n$ repeatedly $K_n=o\left(\epsilon_n^{-1}\right)$ times over the DMC results in average probability of error not exceeding $K_n\times \epsilon_n\triangleq \bar{\epsilon}_N$, and thus one can analyze the performance of the concatenated codebook $\calC_N^{prod}$, whose rate is equal to that of $\calC_n$. 
% 3
\item All the codewords of $\calC_N^{prod}$ lie in a single type-class of sequences of length $N=nK_n$.
% 4
\item Let $X^N$ be the output of the encoder which uses the concatenated codebook $\calC_N^{prod}$ and let $Y^N$ be the output of the DMC channel $W^N$ when fed by $X^N$, then if $q$ is bounded, for all $\Delta>0$ 
\begin{flalign}\label{eq: zeta epsilon2 N hgffgfd}
\mbox{Pr}\left\{ q_N(X^N,Y^N)\geq   \mathbb{E}_{P^{(N)}\times W^{(N)}}(q_N) +\Delta \right\}\leq \epsilon_{2,N},
\end{flalign}
where $ \epsilon_{2,N}\rightarrow 0$ as $N$ tends to infinity.
% 5
\item Let $\tilde{W}^{(N)}$ be a channel which satisfies
\begin{flalign}
& \mbox{Pr}\left\{q_N(X^N,\tilde{Y}^{(N)}) \leq   \mathbb{E}(q_N(X^N,Y^N) )+\Delta \right\} \leq \epsilon_{1,N} \label{eq: conditions 2 first}\\
&P_{Y^N}= P_{\tilde{Y}^N},\label{eq: conditions 3 first}
\end{flalign}
for some $\Delta>0,\epsilon_{1,N} >0$, where $\tilde{Y}^N$ is the output of $\tilde{W}^{(N)}$ when fed by $X^N$. 
Note that (\ref{eq: zeta epsilon2 N hgffgfd})-(\ref{eq: conditions 3 first}) are in fact the conditions (\ref{eq: first condition for corollary2 a})-(\ref{eq: first condition for corollary}) required for Theorem \ref{cr: conditions corollary} to hold, and hence, 
\begin{flalign}
 P_e(\tilde{W}^{(N)},\calC_N^{prod},q_N)\leq & P_e(W^N,\calC_N^{prod},q_N)+ \epsilon_{1,N}+\epsilon_{2,N} .\label{eq: bar epsilon dfn first fhkshvf}
\end{flalign}
% 6
\item From (\ref{eq: bar epsilon dfn first fhkshvf}), it follow that we can invoke Fano's Inequality for the channel $\tilde{W}^{(N)}$ and this yields for all $\Delta>0$, $\epsilon>0$, and $N$ sufficiently large, 
\begin{flalign}\label{eq: W_b first app fdkjbvkjbdfjkv}
R\leq& \min_{\tilde{W}^{(N)}\in \calW_b \left(P^{(N)} ,\epsilon_{N,1}\right) } \frac{1}{N}I(X^N;\tilde{Y}^{N})+R(\bar{\epsilon}_N+\epsilon_{1,N}+\epsilon_{2,N})+\frac{1}{N},
\end{flalign}
where $\left(X^N,\tilde{Y}^N\right)\sim P^{(N)}\times \tilde{W}^{(N)}$, and $\calW_b \left(P^{(n)},\epsilon_{N,1}\right) $ is the set of channels which satisfy (\ref{eq: conditions 2 first})-(\ref{eq: conditions 3 first}).
% 7
\item Let $X^N=(\bX^{(1)},...,\bX^{(K_n)})$ be a blockwise-memoryless source, i.e., $X^n\sim P^{(N)}= \prod_{k=0}^{K_n-1} \bP^{(k)}$. For $\delta>\Delta$, one can show using the law of large numbers that $  \calW_b\left(P^{(N)},\epsilon_{1,N}\right)\supseteq 
\calW_a\left(P^{(N)},\delta\right)
$ where $\epsilon_{1,N}=\frac{B^2}{K_n(\delta-\Delta)}$, $B$ is the bound on $q$ (\ref{eq: bounded q}), and $\calW_a\left(P^{(N)},\delta\right)$ is the set of block-wise memoryless channels $\tilde{W}^{(N)}=\prod_{k=0}^{K_n-1}\tilde{\bW}^{(k)}$ (where $\tilde{\bW}^{(k)}$ is a channel from $\calY^n$ to $\calX^n$ whose input and output are $\bX^{(k)},\tilde{\bY}^{(k)}$, respectively) satisfying the constraints: $\forall k$, $\mathbb{E}_{\bP^{(k)}\times  \tilde{\bW}^{(k)}} q_n\left(\bX^{(k)},\tilde{\bY}^{(k)} \right)$$\geq c^{(k)}+\delta$, $P_{\tilde{\bY}^{(k)}}= P_{\bY^{(k)}}$ 
where $c^{(k)}\triangleq   \mathbb{E}_{\bP^{(k)}\times \bW^{(k)}  }\left(q_n(\bX^{(k)},\bY^{(k)})\right) $.  
Thus, the minimization 
in (\ref{eq: W_b first app fdkjbvkjbdfjkv}) can be upper bounded by minimizing over 
$ \calW_a\left(P^{(N)},\delta\right)$. 
% 8
\item 
Since $\tilde{W}^{(N)}\in \calW_a\left(\prod_{k=0}^{K_n-1} \bP^{(k)},\delta\right)$ is block-wise memoryless, 
$
  \frac{1}{N}I(X^N;\tilde{Y}^N)\leq \sum_{k=0}^{K_n-1} \frac{1}{N}I(\bX^{(k)};\tilde{\bY}^{(k)})$. 
It remains to maximize over the distribution of $X^N$ induced by the concatenated codebook $\calC_N^{prod}$, and hence, the optimizations are in fact performed over block-wise memoryless sources and channels, and this yields
\begin{flalign}
R-\epsilon&\leq \sup_{n}\max_{P^{(n)}} \min_{\tilde{W}^{(n)}\in\calW_n(P^{(n)}, W^{(n)},\delta) }\frac{1}{n} I(X^n;\tilde{Y}^n)\triangleq C_{\bq}(\bW,\delta)
\label{eq: removal of delta fhdskhvkjh}
\end{flalign}
 for all $\epsilon>0$ and $n$ sufficiently large,  where 
\begin{flalign}\label{eq: set of channels def delta hfhbjvbfdsjg}
\calW_n(P^{(n)},W^{(n)},\delta)\triangleq \left\{\tilde{W}^{(n)}:\mathbb{E}_{P^{(n)}\times \tilde{W}^{(n)}}(q_n)\geq \mathbb{E}_{P^{(n)}\times W^{(n)}}(q_n)+\delta ,P_{\tilde{Y}^n}=P_{Y^n}\right\}.
\end{flalign}
\item The case in which the set $\calW_n(P^{(n)},W^{(n)},\delta)$ is empty is treated in Step 9. Finally, we take the limit as $\delta\rightarrow 0$, which yields the desired result.
\end{enumerate}

The difference between $C_q^{(\infty)}(W)$ and the upper bounds given in Theorems \ref{eq: general expression with I} and \ref{th: General formula expression with Is and joint Y Q} is twofold: a) the set over which the minimization is performed and b) the mutual information between $X^n$ and $\tilde{Y}^n$ rather than mutual information density rate of $(X^n,\tilde{Y}^n)$. In fact, it can be shown that the bound given in Theorem \ref{eq: general expression with I} is equal to $C_q^{(\infty)}(W)$ in the case of the DMC $W$, similarly to the proof of Theorem \ref{eq: conjecture proof Theorem}.

We note that the proof of the upper bound of Theorem \ref{eq: conjecture proof Theorem} continues to hold for a larger class of channels rather than DMCs with bounded additive metrics. By inspecting the proof of the upper bound of Theorem \ref{eq: conjecture proof Theorem}, one realizes that the only steps in which the fact that the channel is memoryless and the metric is additive and bounded is used are Steps 2, 4, and 8. 
Roughly speaking, sufficient conditions for establishing a similar upper bound are: \begin{itemize}
\item The channel $\bW=W^{(n)},n\geq 1$ is stationary and has a decaying memory in the sense that if 
$\calC_n$ is an $(n,2^{nR},\epsilon_n)$-code, 
then applying $\calC_n$ repeatedly $K_n=o\left(\epsilon_n^{-1}\right)$ times over the channel $W^{(nK_n)}$ results in average probability of error not exceeding $K_n\times \epsilon_n$.
\item The condition that for every sequence of distributions $\bP$  such that $P^{(n)}\in\calP_{CC}(\calX,n), \forall n$ (see the definition in (\ref{eq: P_CC dfn})) one has\footnote{This is a sufficient for Step 4 to hold. }
\begin{flalign}
p\mbox{-}\liminf \left(q_n(X^n,Y^n)-   \mathbb{E}_{P^{(n)}\times W^{(n)}}(q_n)\right)\leq 0.
\end{flalign}
\item The condition that $C_{\bq}(\bW,\delta)$ (see (\ref{eq: removal of delta fhdskhvkjh}))
is continuous in $\delta$ at the point $\delta=0$. 
\end{itemize}
Under these conditions, the mismatch capacity is upper bounded by $\lim_{\delta\rightarrow 0}C_{\bq}(\bW,\delta)$. 

Consider the case in which $\bq$ is such that for all $n,x^n,y^n$, the function $q_n(x^n,y^n)$ depends only on the joint type-class of $x^n$ and $y^n$ (\cite{CsiszarKorner81graph}).  In this case, we denote with a little abuse of notation $q_n(x^n,y^n)=q_n(\hat{P}_{x^n,y^n})$, where $\hat{P}_{x^n,y^n}$ is the joint empirical distribution induced by $(x^n,y^n)$.  We shall refer to $\bq$ that satisfies this condition as a type-dependent $\bq$. 
An important thing to notice is that when $\bq$ is type-dependent, the mismatch capacity takes a special form. Note that
\begin{flalign}\label{eq: method of types chain}
\Phi_{q_n}(c,\mu,\by)=&\sum_{\tilde{\bx}\in\calX^n:q_n(\tilde{\bx},\by)\geq c}\mu(\tilde{\bx})\nonumber\\
=&\sum_{\hat{P}_{\tilde{\bx}|\by} : q_n(\hat{P}_{\tilde{\bx},\by} )\geq c}\mu(T(\hat{P}_{\tilde{\bx}|\by}))\nonumber\\
\doteq&\max_{\hat{P}_{\tilde{\bx}|\by} : q_n(\hat{P}_{\tilde{\bx},\by} )\geq c}\mu(T(\hat{P}_{\tilde{\bx}|\by})),
\end{flalign}
where the last equality on the exponential scale follows since the number of types is polynomial in $n$. 
Moreover, from Step $1$ of the proof of Theorem \ref{eq: conjecture proof Theorem}, it follows that without loss of asymptotic optimality, one can assume that all the codewords of the codebook lie in a single type-class. This yields the following corollary. 
\begin{corollary}\label{th: General formula expression Method of types}
The mismatch capacity of the finite alphabet channel $\bW$ with type-dependent decoding metrics sequence $\bq$ is given by
\begin{flalign}\label{cr: IGMILM 1}
&C_{\bq}(\bW)
=  \sup_{\bP\in\calP^{(\infty)}} p\mbox{-}\liminf  -\frac{1}{n}\log\left(\max_{\tilde{P}\in\calP_n(\calX\times\calY):\tilde{P}_Y=\hat{P}_{Y^n} , q_n(\tilde{P} )\geq q_n(\hat{P}_{X^n,Y^n} )}P^{(n)} (T(\tilde{P}_{X,Y}|Y^n) )\right)\nonumber\\
= & \sup_{\bP:\; \forall n, P^{(n)}\in\calP_{CC}(\calX,n)} p\mbox{-}\liminf  -\frac{1}{n}\log\left(\max_{\tilde{P}\in\calP_n(\calX\times\calY):\tilde{P}_Y=\hat{P}_{Y^n} , \tilde{P}_X=\hat{P}_{X^n}, q_n(\tilde{P} )\geq q_n(\hat{P}_{X^n,Y^n} )}P^{(n)} (T(\tilde{P}_{X,Y}|Y^n) )\right),
\end{flalign}
where $(X^n,Y^n)\sim P^{(n)}\times W^{(n)}$, $T(\tilde{P}_{X,Y}|Y^n)$ is the set of $x^n\in\calX^n$ whose joint empirical statistics with $Y^n$ is given by $\tilde{P}_{X,Y}$, and $\hat{P}_{X^n,Y^n}$ is the empirical distribution induced by $(X^n,Y^n)$. 
\end{corollary}

Next, we deduce the following identity in the DMC case where $W^{(n)}=W^n$ and $q_n$ is given in (\ref{eq: additive decoder decision rule}).
\begin{lemma}\label{lm: identity C infty Lemma}
The following identity holds for a DMC $W$ with a bounded additive decoding metric $q_n$ (\ref{eq: additive decoder decision rule})
\begin{flalign}\label{eq: identity C infty Lemma}
C_q^{(\infty)}(W) = \sup_{\bP\in \calP^{(\infty)}}\sup_{\epsilon>0} p\mbox{-}\liminf -\frac{1}{n}\log\left(\Phi_{q_n}(\mathbb{E}(q_n(X^n,Y^n))-\epsilon, P^{(n)}, Y^n)\right),
\end{flalign}
where $C_q^{(\infty)}(W)$ is defined in (\ref{eq: dfn Cq infty W}), $(X^n,Y^n)\sim P^{(n)}\times W^{(n)}$ and the supremum can be restricted to $\{\bP\in \calP_{U}^{(\infty)}\}$. 
\end{lemma}
\begin{proof}
From Theorem \ref{th: .ksjdf .jdfhkv.jh} it follows that the r.h.s.\ of (\ref{eq: identity C infty Lemma}) is the rate achievable by  threshold decoding with threshold level $\tau_n^*=\mathbb{E}_{P^{(n)}\times W^n}(q_n(X^n,Y^n))-\epsilon$ for $\epsilon>0$ arbitrarily small. 
Lemma \ref{lm: threshold is inferior} and Theorem \ref{eq: conjecture proof Theorem} imply that the r.h.s.\ of (\ref{eq: identity C infty Lemma}) is lower than $C_q^{(\infty)}(W)$. 

To prove the opposite inequality, we pick an empirical distribution $P\in\calP_n(\calX^n)$ and use random coding uniform in $T(P)$, i.e., $P^{(n)}(x^n)=\frac{1\{x^n\in T(P)\}}{|T(P)|}$. Assume without loss of generality that the transmitted message $S$ is $i$ and denote the random codebook $\calC=\{X^n(m)\}_{m=1}^{M_n}$ . Fix an arbitrarily small $\epsilon>0$, and set $\tau_n^*=\mathbb{E}_{P^{(n)}\times W^{(n)}}(q_n(X^n,Y^n))-\epsilon$. Let $X^n$ and $Y^n$ denote the input and output of the channel $W^{(n)}$, respectively. 
On one hand, the ensemble average probability of error, denoted $P_e(W^{(n)},(q_n,\tau_n^*))$, satisfies 
\begin{flalign}\label{eq: tau star eq 1}
&P_e(W^{(n)},(q_n,\tau_n^*)) \nonumber\\
= & \mathbb{E}\left(1-\left(1-\Phi_{q_n}(\tau_n^*,P^{(n)},Y^n)\right)^{M_n}\right)\nonumber\\
\stackrel{(a)}{\geq} & \frac{1}{2}\mathbb{E}\min\left\{1,M_n\Phi_{q_n}(\tau_n^*,P^{(n)},Y^n)\right\}
\nonumber\\
\geq & \frac{1}{2} \mbox{Pr}\left\{M_n\Phi_{q_n}(\tau_n^*,P^{(n)},Y^n)>1\right\},
\end{flalign}
where $(a)$ follows from \cite[Lemma 1]{SomekhBaruchMerhav07} stating that for $a\in[0,1]$, one has 
\begin{flalign}\label{eq: doteq Lemma}
\frac{1}{2}\min\{1,Ma\}\leq 1-\left[1-a\right]^M\leq \min\{1,Ma\}.
\end{flalign}
On the other hand assuming without loss of generality that the transmitted message is $i$, we have 
\begin{flalign}\label{eq: tau star eq 2}
&P_e(W^{(n)},(q_n,\tau_n^*))\nonumber\\
=& \mbox{Pr}\left( \left\{q_n(X^n,Y^n)<\tau_n^*\right\} \cup \left\{ \exists j\neq i:\; q_n(X^n(j),Y^n)\geq \tau_n^*, q_n(X^n,Y^n)\geq \tau_n^* \right\}\right)\nonumber\\
\leq& \mbox{Pr}\left( \left\{q_n(X^n,Y^n)<\tau_n^*\right\} \right) +\mbox{Pr}\left( \exists j\neq i:\; q_n(X^n(j),Y^n)\geq \tau_n^*, q_n(X^n,Y^n)\geq \tau_n^* \right) \nonumber\\
=& \mbox{Pr}\left( \left\{q_n(X^n,Y^n)<\tau_n^*\right\} \right) +\mbox{Pr}\left( \exists j\neq i:\; q_n(X^n(j),Y^n)\geq \tau_n^*, q_n(X^n,Y^n)\in[ \tau_n^*,\tau_n^*+2\epsilon]  \right)\nonumber\\
& +\mbox{Pr}\left( \exists j\neq i:\; q_n(X^n(j),Y^n)\geq \tau_n^*, q_n(X^n,Y^n)\geq \tau_n^*+2\epsilon  \right) \nonumber\\
\leq& \mbox{Pr}\left( \left\{q_n(X^n,Y^n)<\tau_n^*\right\} \right) +\mbox{Pr}\left( \exists j\neq i:\; q_n(X^n(j),Y^n)\geq q_n(X^n,Y^n)-2\epsilon \right)\nonumber\\
& +\mbox{Pr}\left(  q_n(X^n,Y^n)\geq \tau_n^*+2\epsilon  \right) .
\end{flalign}
Now, since $P^{(n)}$ is uniform in a single type-class, $q_n(X^n,Y^n)-\mathbb{E}_{P^{(n)}\times W^{(n)}}(q_n)$ converges in probability to $0$ (see (\ref{eq: the equation of Claim 1}),  Claim \ref{eq: DMC convergence claim}). Hence, by definition of $\tau_n^*$, both $\mbox{Pr}\left( \left\{q_n(X^n,Y^n)<\tau_n^*\right\} \right) $ and $\mbox{Pr}\left(  q_n(X^n,Y^n)\geq \tau_n^*+2\epsilon  \right)$ tend to zero as $n$ tends to infinity. 
Consequently from (\ref{eq: tau star eq 1}),(\ref{eq: tau star eq 2}), for all $\delta>0$ there exists $n(\delta)$ such that for all $n\geq n(\delta)$,
\begin{flalign}\label{eq: tau star eq 3}
\frac{1}{2} \mbox{Pr}\left\{M_n\Phi_{q_n}(\tau_n^*,P^{(n)},Y^n)>1\right\} \leq \mbox{Pr}\left( \exists j\neq i:\; q_n(X^n(j),Y^n)\geq q_n(X^n,Y^n)-2\epsilon \right)+\delta,
\end{flalign}
Now, the term $\mbox{Pr}\left( \exists j\neq i:\; q_n(X^n(j),Y^n)\geq q_n(X^n,Y^n)-2\epsilon \right)$ is the ensemble average probability of error using the mismatched decoder metric $q'(x,y)=q(x,y)-2\epsilon$. Since $q_n(X^n,Y^n)$ converges in probability to $\mathbb{E}_{P^{(n)}\times W^{(n)}}(q_n)$, and the empirical distribution of $Y^n$, $\hat{P}_{Y^n}$, converges in probability to $P_Y$, for $\Delta>0$ arbitrarily small, with probability approaching $1$, $q_n(X^n,Y^n)\geq \mathbb{E}_{PW}(q)-\Delta$ and there exists a vanishing sequence $\zeta_n$ such that 
\begin{flalign}\label{eq: tau star eq 4}
& \mbox{Pr}\left( \exists j\neq i:\; q_n(X^n(j),Y^n)\geq q_n(X^n,Y^n)-2\epsilon \right)\nonumber\\
\leq &M_n\cdot(n+1)^{|\calX|(|\calY|+1)}2^{-n\min_{\tilde{P}_{X,Y}:\; \mathbb{E}_{\tilde{P}}(q)\geq  \mathbb{E}_{PW}(q)-2\epsilon-\Delta,\tilde{P}_X=P_X,\tilde{P}_Y=P_Y}I_{\tilde{P}}(X;Y)}+\zeta_n\nonumber\\
=&2^{n\left(R-C_{q-2\epsilon-\Delta}^{(1)}(W)+|\calX|(|\calY|+1)\frac{\log(n+1)}{n}\right)}+\zeta_n,
\end{flalign}
where the inequality follows from the union bound and because the number of joint type-classes over $\calX\times\calY$ is upper bounded by $(n+1)^{|\calX||\calY|-1}$, and $C_q^{(1)}(W)$ is defined in (\ref{eq: Cq1 dfn}). 
In light of (\ref{eq: tau star eq 3}) and (\ref{eq: tau star eq 4}), if for some $\epsilon_n>0$ one has $R<C_{q+2\epsilon-\Delta}^{(1)}(W)-|\calY|(|\calX|+1)\frac{\log(n+1)}{n}$ the term $ \mbox{Pr}\left\{M_n\Phi_{q_n}(\tau_n^*,P^{(n)},Y^n)>1\right\} $ must vanish, and thus the r.h.s.\ of (\ref{eq: identity C infty Lemma}) is upper bounded by $C_{q-2\epsilon-\Delta}^{(1)}(W)$. 
In other words, by continuity of $C_q^{(1)}(W)$  in $q$, the threshold decoder achieves a rate arbitrarily close to $C_q^{(1)}(W)$. The same argument holds for $C_q^{(K)}(W), K>1$ (with slower rate of convergence), and thus, since $\lim_{K\rightarrow\infty}C_q^{(K)}(W)=C_q^{(\infty)}(W)$ the threshold decoder achieves a rate arbitrarily close to $C_q^{(\infty)}(W)$ for $n$ sufficiently large.
\end{proof}
We note that for a similar argument that was used in (\ref{cr: IGMILM 1}), the right hand side of (\ref{eq: identity C infty Lemma}) is equal to
\begin{flalign}\label{eq: identity C infty Lemma 2}
 \sup_{\bP:\; \forall n, P^{(n)}\in \calP_{CC}(\calX,n)}\sup_{\epsilon>0} p\mbox{-}\liminf -\frac{1}{n}\log\left(\Phi_{q_n}(\mathbb{E}_{\hat{P}_{X^n}\times W}(q_n(X,Y))-\epsilon, P^{(n)}, Y^n)\right).
\end{flalign}

Theorem \ref{eq: conjecture proof Theorem} reinforces Csisz\'{a}r and Narayan's conjecture and thus solves Open Problem $6$ in \cite[Section 5]{CsiszarNarayan95} as well as two additional problems: 
\begin{itemize}
\item Open Problem 5 
\cite[Section 5]{CsiszarNarayan95}: It is stated that $C_q(W)=C_q^{(\infty)}(W)$ is a sufficient condition for the existence of 
rate approaching $R<C_q(W)$ codes and probability of error decaying
to zero exponentially as the block length goes to infinity, and due to Theorem \ref{eq: conjecture proof Theorem}, this is indeed the case.
\item Open Problem 7 
of \cite{CsiszarNarayan95} concerns the threshold $q$-capacity, which, as states in \cite{CsiszarNarayan95}, is clearly upper bounded by the mismatch $q$-capacity (see Lemma \ref{lm: threshold is inferior} for the proof in the general channel case).
Further, $C_q^{(\infty)}(W)$ constitutes a
lower bound on the threshold $q$-capacity too (\cite{CsiszarKorner81graph},\cite{Hui83}). 
Hence, as states in \cite{CsiszarNarayan95}, the affirmative answer to Problem 6 implies that the constant threshold capacity of the DMC is equal to the mismatch capacity. This claim is also proved in Lemma \ref{lm: identity C infty Lemma} (see (\ref{eq: identity C infty Lemma})). 
\end{itemize}
As mentioned before, Theorem \ref{eq: conjecture proof Theorem} was proved in \cite{CsiszarNarayan95} for the special case of erasures-only capacity of the DMC. 
We noted that the erasures only capacity can be considered as the supremum of rates achievable by decoding with respect to the metric $q_n^{eo}$ (see  (\ref{eq: eo metric for vectors})). In the DMC case, this metric is equivalent (capacity-wise) to the additive metric
\begin{flalign}
q_{n,eo}(x^n,y^n)=\sum_{i=1}^n 1\{W(y_i|x_i)>0\},\label{eq: eo metric for scalars}
\end{flalign}
i.e., with $q(x,y)= 1\{W(y|x)>0\}$,
and similarly to the metric $q_n^{eo}$, in this special case, it always holds that the actual transmitted codeword $X^n$ and received signal $Y^n$ satisfy
\begin{flalign}
q_{n,eo}(X^n,Y^n)=1.
\end{flalign}
Hence, the erasures-only capacity $C_{eo}(W)$ is equal to the threshold capacity with threshold level\footnote{In this case, one need not take a threshold level $\tau$ arbitrarily close to $1$ but rather set $\tau=1$ since the equality (\ref{eq: eo metric for scalars}) holds surely, for every realization of $(X^n,Y^n)$.} $\tau=1$  (see (\ref{eq: VerduHan UB threshold})), and  Lemma \ref{lm: identity C infty Lemma} applied to this metric in the DMC case yields an alternative expression to the erasures-only capacity of the DMC $W$, $C_{eo}(W)$, which is given by 
\begin{flalign}
C_{q_{eo}}^{(\infty)}(W)=&\sup_{\bP\in \calP^{(\infty)}}p\mbox{-}\liminf -\frac{1}{n}\log \left( \Phi_{q_{n,eo}}(1,P^{(n)},Y^n) \right)\nonumber\\
=&\sup_{\bP\in \calP^{(\infty)}}p\mbox{-}\liminf -\frac{1}{n}\log \left(\sum_{\tilde{x}^n:\; W^n(Y^n|\tilde{x}^n)>0}P^{(n)}(\tilde{x}^n)\right),
\end{flalign}
where the supremum can be restricted to $\{\bP\in \calP_{U}^{(\infty)}\}$. 

We note that using Corollary \ref{cr: Upper Bound Expectation} (see (\ref{eq: the General formula ksdjhfkjhsk 22}) and (\ref{eq: inequality between liminf and expectation})) we obtain the bound  
 \begin{flalign}\label{eq: the General formula ksdjhfkjhsk 22 erasures}
C_{q_{eo}}^{(\infty)}(W) \leq & \sup_{\bP\in \calP_{U}^{(\infty)}}    \liminf_{n\rightarrow\infty} \frac{1}{n} \mathbb{E}\left(
\log\frac{1}{\sum_{\tilde{x}^n:\; W^n(Y^n|\tilde{x}^n)>0}P^{(n)}(\tilde{x}^n)}\right),
\end{flalign}
which was recently established in \cite[Theorem 2.1]{BunteLapidothSamorodintskyArxiv2013}, and as pointed out in \cite{BunteLapidothSamorodintskyArxiv2013}, 
it is easy to show that the right hand side of (\ref{eq: the General formula ksdjhfkjhsk 22 erasures}) with the $\liminf$ replaced by $\limsup$ is an achievable rate as the $n$-letter extension of Forney's \cite{Forney1968Erasure} lower bound on $C_{eo}(W)$,  and therefore the above inequality holds with equality, i.e.,
 \begin{flalign}\label{eq: the General formula ksdjhfkjhsk 22 erasures equality}
C_{q_{eo}}^{(\infty)}(W) = & \sup_{\bP\in \calP_{U}^{(\infty)}}    \lim_{n\rightarrow\infty} \frac{1}{n} \mathbb{E}\left(
\log\frac{1}{\sum_{\tilde{x}^n:\; W^n(Y^n|\tilde{x}^n)>0}P^{(n)}(\tilde{x}^n)}\right).
\end{flalign}

\section{Random Coding Over a Given Codebook}\label{sc: The Random Coding Over a Given Codebook}

In this section, we establish a connection between the maximal probability of erroneously decoding a message using a codebook $\calC_n$ and a decoder $q_n$ and the average probability of the random codebook, of slightly lower rate, whose codewords are drawn i.i.d.\ over $\calC_n$ using the same decoder. This derivation, beyond being interesting in itself,  enables to establish an alternative proof of Theorem \ref{th: General formula expression}, which is based on the analysis of the performance of a random code. Along with the direct part of the proof of Theorem \ref{th: General formula expression} it constitutes a complete proof which is based on random coding.

 Recall that $P_e(W^{(n)},\calC_n,q_n)$ denotes the average probability of error incurred using the $(M_n,n)$-code $\calC_n$ with the decoder $q_n$ (\ref{eq: decoder decision rule}) 
where $\by$ is the output of the channel $W^{(n)}$. Let 
\begin{flalign}
 P_{\max}(W^{(n)},\calC_n,q_n)=\max_{i\in\{1,...,M_n\}}\mbox{Pr}\left(\hat{S}\neq S| S=i\right)
 \end{flalign}
 be the maximal probability of error incurred by the same code and decoder.

Denote by $\bar{P}_e(W^{(n)}, \calC_n,q_n,R-\epsilon)$ the average probability of error of the random code of size
\begin{flalign}\label{eq: Mepsilon dfn}
M_{n,\epsilon}&=2^{n(R-\epsilon)}
\end{flalign}
with independent codewords, each drawn uniformly over $\calC_n$ using the same decoder $q_n$. 

We next present a lemma which implies that $\bar{P}_e(W^{(n)}, \calC_n,q_n,R-\epsilon)$ 
can be upper bounded by $P_{\max}(W^{(n)},\calC_n,q_n)$ up to a vanishing quantity.

 \begin{lemma}\label{lm: essential lemma ldfhkauh}
 If $\calC_n$ is a codebook composed of $M_n=2^{nR}$ distinct codewords and $\epsilon\in (0,R)$, then 
\begin{flalign}
\bar{P}_e(W^{(n)},\calC_n,q_n,R-\epsilon)\leq  
P_{\max}(W^{(n)},\calC_n,q_n)+ \delta_n.
\end{flalign}
where
\begin{flalign}
\delta_n\triangleq& \frac{1}{2}\cdot (2^{-n\epsilon}-2^{-nR})
\end{flalign}
\end{lemma}
\begin{proof}
Let $\calC_n=\{\bx_1,...,\bx_{M_n}\}$ be the given codebook. 
Let $q_n$ be a decoder of the form (\ref{eq: decoder decision rule}). 
Next draw $M_{n,\epsilon}=2^{n(R-\epsilon)}$ codewords independently using $P^{(n)}\in\calP(\calX^n)$, and assign indices to the drawn codewords by their order of appearance. The codewords constitute the random codebook $\tilde{\calC}_n=\left\{\tilde{\bX}_1,...,\tilde{\bX}_{M_{n,\epsilon}}\right\}$.  
Let $\calB$ be the random set of codewords which appear only once in $\tilde{\calC}_n$, let $X^n$ be the transmitted random codeword (uniformly distributed over $\tilde{\calC}_n$), and let $S$ and $\hat{S}$ stand for the transmitted message and the output of the decoder $q_n$, respectively. 

One has
\begin{flalign}
&\bar{P}_e(W^{(n)},\calC_n,q_n,R-\epsilon) \nonumber\\
=& \mbox{Pr}\left(\hat{S}\neq S, X^n\in \calB\right)+ \mbox{Pr}\left(\hat{S}\neq S, X^n\in \calB^c\right) \nonumber\\
=& \mbox{Pr}\left( X^n\in \calB\right)\mbox{Pr}\left(\hat{S}\neq S| X^n\in \calB\right)+ \mbox{Pr}\left( X^n\in \calB^c\right) \mbox{Pr}\left(\hat{S}\neq S| X^n\in \calB^c\right) \nonumber\\
\stackrel{(a)}{=}&(1-1/M_n)^{M_{n,\epsilon}-1}\cdot \mbox{Pr}\left(\hat{S}\neq S| X^n\in \calB\right)+  [1-(1-1/M_n)^{M_{n,\epsilon}-1}]\cdot \mbox{Pr}\left(\hat{S}\neq S| X^n\in \calB^c\right) \nonumber\\
\leq & \mbox{Pr}\left(\hat{S}\neq S| X^n\in \calB\right)+  [1-(1-1/M_n)^{M_{n,\epsilon}-1}]\cdot \mbox{Pr}\left(\hat{S}\neq S| X^n\in \calB^c\right) \end{flalign}
where $(a)$ follows since $\mbox{Pr}\left( X^n\in \calB\right)=\sum_{\bx\in\calC}\frac{1}{M_n}\mbox{Pr}\left(\cap_{i=2}^{M_{n,\epsilon}}\{ \tilde{\bX}_i\neq \bx\} \right)=\sum_{\bx\in\calC_n}\frac{1}{M}\prod_{i=2}^{M_{n,\epsilon}} \left(1-1/M_n\right) $.

Now, note that  

\begin{flalign}
&1-(1-1/M_n)^{M_{n,\epsilon}-1}\nonumber\\
\geq &\frac{1}{2}\min\left\{1,(M_{n,\epsilon}-1)/M_n\right\}\nonumber\\
= &\frac{1}{2}\cdot (2^{-n\epsilon}-2^{-nR}),
\end{flalign}
where the inequality follows from (\ref{eq: doteq Lemma}). 
Therefore, we obtain
\begin{flalign}
&\bar{P}_e(W^{(n)},\calC_n,q_n,R-\epsilon) \nonumber\\
\leq &   
\mbox{Pr}\left(\hat{S}\neq S| X^n\in \calB\right)+  \delta_n \cdot \mbox{Pr}\left(\hat{S}\neq S| X^n\in \calB^c\right) \nonumber\\
\stackrel{(a)}{\leq} &  
 P_{\max}(W^{(n)},\calC_n,q_n)+  \delta_n
 \end{flalign}
where $(a)$ follows because given $X^n\in \calB$, it is known that $X^n$ appears only once in the codebook, while other codewords may appear more than once, and
using the decoder $q_n$ of the form (\ref{eq: decoder decision rule}) with a subset of the codewords of the original codebook enlarges the decision regions and cannot increase the maximal probability of error.

\end{proof}

 Having proved Lemma \ref{lm: essential lemma ldfhkauh}, we can establish an alternative proof of the converse part of Theorem \ref{th: General formula expression}.

{\bf An Alternative Proof of the Converse Part of Theorem \ref{th: General formula expression}}

\begin{proof}
As a result of Lemma \ref{lm: essential lemma ldfhkauh}, to obtain an upper bound on the capacity (w.r.t.\ maximal probability of error), we can analyze the average probability of error of the average code whose codewords are drawn uniformly over a codebook $\calC_n$.
Let $P^{(n)}\in\calP(\calX^n)$ be uniform over $\calC_n$. 
Clearly, for sufficiently large $n$
\begin{flalign}
& P_{\max}(W^{(n)},\calC_n,q_n)+\delta_n\nonumber\\
\stackrel{(a)}{\geq} &\bar{P}_e(W^{(n)},\calC_n,q_n,R-\epsilon)\nonumber\\
\stackrel{(b)}{=}  &
\mathbb{E}\left[1-\left(1-\Phi_{q_n}(q_n(X^n,Y^n),P^{(n)},Y^n)\right)^{(M_{n,\epsilon}-1)}\right] \nonumber\\
\geq& 
\mathbb{E}\left(\left[1-\left(1-\Phi_{q_n}(q_n(X^n,Y^n),P^{(n)},Y^n)\right)^{(M_{n,\epsilon}-1)}\right]1\{(M_{n,\epsilon}-1) \cdot \Phi_{q_n}(q_n(X^n,Y^n),P^{(n)},Y^n)\geq e^{n\gamma}\}\right)\nonumber\\
\geq& 
\mathbb{E}\left(\left[1-\left(1-\frac{e^{n\gamma}}{M_{n,\epsilon}-1}\right)^{(M_{n,\epsilon}-1)}\right]1\{(M_{n,\epsilon}-1) \cdot \Phi_{q_n}(q_n(X^n,Y^n),P^{(n)},Y^n)\geq e^{n\gamma}\}\right)\nonumber\\
\stackrel{(c)}{\geq}& 
\left[1-e^{-e^{n\gamma}}\right]\mathbb{E} \left(1\{(M_{n,\epsilon}-1) \cdot \Phi_{q_n}(q_n(X^n,Y^n),P^{(n)},Y^n)\geq e^{n\gamma}\}\right)\nonumber\\
\geq &\mbox{Pr}\left\{(M_{n,\epsilon}-1)\Phi_{q_n}(q_n(X^n,Y^n),P^{(n)},Y^n)\geq e^{n\gamma}\right\}-e^{-e^{n\gamma}},
\end{flalign}
where 
$(a)$ follows from Lemma \ref{lm: essential lemma ldfhkauh}, 
$(b)$ follows since the decoder successfully decodes $S=m$ only if $q_n(\bx(m),\by)> q_n(\bx(j), \by)$ for all $j\neq m$ and an error occurs if there is at least one "failure" in $M_{n,\epsilon}-1$ Bernoulli experiments, and $(c)$ follows from the inequality $(1-L^{-1})^L\leq e^{-1}$ applied to $L=\frac{M_{n,\epsilon}-1}{e^{n\gamma}}$.

Next, for all $\delta>0$, there exists sufficiently large $n$ such that, $M_{n,\epsilon}-1\geq M_{n,\epsilon} 2^{-n\delta}= M 2^{-n(\epsilon+\delta)}$, thus for sufficiently large $n$, 
\begin{flalign}
& \mbox{Pr}\left\{(M_{n,\epsilon}-1)\Phi_{q_n}(q_n(X^n,Y^n),P^{(n)},Y^n)\geq e^{n\gamma}\right\}\nonumber\\
=& \mbox{Pr}\left\{-\frac{1}{n}\log\left(\Phi_{q_n}(q_n(X^n,Y^n),P^{(n)},Y^n)\right)\leq \frac{1}{n}\log(M_{n,\epsilon}-1) -\gamma \right\}\nonumber\\
\geq& \mbox{Pr}\left\{-\frac{1}{n}\log\left(\Phi_{q_n}(q_n(X^n,Y^n),P^{(n)},Y^n)\right)\leq \frac{1}{n}\log(M_n)-\epsilon-\delta -\gamma\right\}.
\end{flalign}
To conclude, we have for sufficiently large $n$, a weaker version of (\ref{eq: VerduHan UB}):
\begin{flalign}
& P_{\max}(W^{(n)},\calC_n,q_n)+\delta_n\nonumber\\
\geq & \mbox{Pr}\left\{-\frac{1}{n}\log\left(\Phi_{q_n}(q_n(X^n,Y^n),P^{(n)},Y^n)\right)\leq \frac{1}{n}\log(M)-\epsilon-\delta -\gamma\right\}-e^{-e^{n\gamma}} .\label{eq: we have for suff}
\end{flalign}
It is easy to realize that since $\delta_n\rightarrow 0$ as $n$ tends to infinity, by definition of the limit inferior in probability, if 
\begin{flalign}
R> \sup_{\bP}  p\mbox{-}\liminf  -\frac{1}{n}\log\left(\Phi_{q_n}(q_n(X^n,Y^n),P^{(n)},Y^n)\right)= \underline{K}_{\bq}(\bW),
\end{flalign} the maximal probability of error does not vanish as $n$ tends to infinity. This part of the proof follows similarly to equations (3.11)-(3.14) in \cite{VerduHan1994} as follows: 
We show that the assumption that $R=R_0= \underline{K}_{\bq}(\bW)+3\gamma+\delta+\epsilon$ is achievable leads to a contradiction for arbitrarily small positive $\gamma,\epsilon,\delta$. 
Since by assumption $R_0$ is achievable, for all $\zeta>0$ and $n$ sufficiently large, there exists an $(n,M_n,\epsilon_n)$-code satisfying 
\begin{flalign}\label{eq: 18}
\liminf_{n\rightarrow\infty} \frac{1}{n}\log M_n\geq R_0-\zeta
\end{flalign}
and $\lim_{n\rightarrow\infty}\epsilon_n=0$.  Define $X^n$ as the random variable uniformly distributed over the code, and $Y^n$ the corresponding output, we 
then have from (\ref{eq: we have for suff}) for all $\gamma>0$ and sufficiently large $n$, 
\begin{flalign}
\mbox{Pr}\left\{-\frac{1}{n} \log \left(\Phi_{q_n}( q_n(X^n,Y^n),P^{(n)},Y^n)\right) \leq \frac{1}{n}\log M_n-\epsilon-\delta -\gamma \right\}\leq  \epsilon_n+\delta_n+e^{-e^{n\gamma}}.
\end{flalign}
On the other hand from (\ref{eq: 18}) for $n$ sufficiently large, it holds that
\begin{flalign}
\mbox{Pr}\left\{-\frac{1}{n} \log \left(\Phi_{q_n}( q_n(X^n,Y^n),P^{(n)},Y^n)\right) \leq R_0-\epsilon-\delta-\gamma-\zeta \right\}\leq \epsilon_n+\delta_n+e^{-e^{n\gamma}}.
\end{flalign}
Now, pick $\zeta=\gamma$ and substitute $R_0=\underline{K}_{\bq}(\bW)+3\gamma+\epsilon+\delta$, this yields 
\begin{flalign}
\mbox{Pr}\left\{-\frac{1}{n} \log \left(\Phi_{q_n}( q_n(X^n,Y^n),P^{(n)},Y^n)\right) \leq \underline{K}_{\bq}(\bW)+\gamma \right\}\leq  \epsilon_n+\delta_n+e^{-e^{n\gamma}}.
\end{flalign}
However, the definition of $\liminf$ in probability implies the existence of $\epsilon>0$ such that 
for infinitely many $n$'s satisfying 
\begin{flalign}
\mbox{Pr}\left\{-\frac{1}{n} \log \left(\Phi_{q_n}( q_n(X^n,Y^n),P^{(n)},Y^n)\right) \leq \underline{K}_{\bq}(\bW)+\gamma \right\}\geq \epsilon,
\end{flalign}
since $\lim_{n\rightarrow\infty}\epsilon_n=0$ this yields a contradiction.

Finally, by observing that the mismatch capacity w.r.t.\ maximal probability of error is equal to the mismatch capacity w.r.t.\ average probability of error, this concludes the second proof of the converse part of Theorem \ref{th: General formula expression}.

\end{proof}

\section{Conditions for the Existence of a Strong Converse}\label{sc: Strong Converse}
In this section, we present a condition which is necessary and sufficient for a channel with mismatched decoding to have a strong converse. 
We say that a channel satisfies the strong converse property if for all $\delta>0$, every sequence of $(M_n,n,\epsilon_n)$-codes with $\frac{1}{n}\log M_n>C_{\bq}(\bW)+\delta , \;\forall n$,  satisfies $\lim_{n\rightarrow\infty}\epsilon_n=1$.
Recall the definition of $\overline{K}_{\bq}(\bW)$ (\ref{eq: overline K dfn}). \begin{theorem}\label{th: condition STRONG}
A channel $\bW$ satisfies the strong converse iff 
\begin{flalign}\label{eq: condition STRONG}
\overline{K}_{\bq}(\bW)=\underline{K}_{\bq}(\bW).
\end{flalign}
\end{theorem}
The necessary and sufficient condition for the existence of a strong converse which was established by Verd\'{u} and Han \cite{VerduHan1994} in the matched metric case is 
\begin{flalign}\label{eq: jfbdkbnvk.dfjnk.}
\sup_{\bP}\bar{I}(X^n;Y^n)=\sup_{\bP}\underline{I}(X^n;Y^n),
\end{flalign}
and thus, (\ref{eq: condition STRONG}) for $\bq=\bW$ is equivalent to (\ref{eq: jfbdkbnvk.dfjnk.}). This is not surprising in light of (\ref{eq: identity equation a}).
\begin{proof}
{\bf Sufficiency}: To prove sufficiency of condition (\ref{eq: condition STRONG}), we assume a sequence of codes $\{\calC_n\}_{n\geq 1}$ is given such that 
\begin{flalign}
\frac{1}{n}\log(M_n)\geq \underline{K}_{\bq}(\bW)+\delta.
\end{flalign}
Let $P^{(n)}$ be the uniform distribution over $\calC_n$, and $(X^n,Y^n)\sim P^{(n)}\times W^{(n)}$. Now, from Lemma \ref{lm: VerduHan Lemma}, we have
\begin{flalign}\label{eq: dfkjvbjk}
&P_e(W^{(n)},\calC_n,q_n)\nonumber\\
=& \mbox{Pr}\left\{-\frac{1}{n} \log \left(\Phi_{q_n}(q_n(X^n,Y^n),P^{(n)},Y^n)\right) < \frac{1}{n}\log M_n \right\}\nonumber\\
\geq & \mbox{Pr}\left\{-\frac{1}{n} \log \left(\Phi_{q_n}(q_n(X^n,Y^n),P^{(n)},Y^n)\right) <\underline{K}_{\bq}(\bW)+\delta \right\}\nonumber\\
= & \mbox{Pr}\left\{-\frac{1}{n} \log \left(\Phi_{q_n}(q_n(X^n,Y^n),P^{(n)},Y^n)\right) <\overline{K}_{\bq}(\bW)+\delta \right\}
\end{flalign}
where the last step follows from (\ref{eq: condition STRONG}). By definition of the limit superior in probability and by the fact that we take the supremum over $\bP$ in the definition of $\overline{K}_{\bq}(\bW)$, the r.h.s.\ of (\ref{eq: dfkjvbjk})  must go to $1$ as $n$ tends to infinity, and hence so must the l.h.s.\ and the channel has a strong converse.

{\bf Necessity}: To prove necessity of condition (\ref{eq: condition STRONG}), we assume that the channel satisfies the strong converse property. 
Let $\calG(n,M_n)$ denote the set of $(n,M_n)$-codebooks over $\calX^n$.  
Denote $M_n^{\Delta}\triangleq \left\lceil 2^{n( \underline{K}_{\bq}(\bW)+\Delta)}\right\rceil$. Thus, for all $\Delta>0$, \begin{flalign}\label{eq: e-epsilon eq}
\liminf_{n\rightarrow \infty} \underset{\calC_n\in \underset{M_n\geq M_n^{\Delta}}{\cup}\calG(n,M_n)}{\inf} P_e(W^{(n)},\calC_n,q_n)=1.
\end{flalign}
Clearly, one can assume that the infimum is attained within $\calG(n,M_n^{\Delta})$ without loss of generality, since for every $(n,M_n)$-codebook $\calC_n$ and every $M_n'<M_n$, there exists a sub-codebook $\calC_n'$ of size $M_n'$ such that $P_e(W^{(n)},\calC_n',q_n)\leq P_e(W^{(n)},\calC_n,q_n)$. Hence, for all $\Delta>0$,
\begin{flalign}\label{eq: e-epsilon eq bb}
\liminf_{n\rightarrow \infty} \underset{\calC_n\in \calG(n,M_n^{\Delta})}{\inf} P_e(W^{(n)},\calC_n,q_n)=1.
\end{flalign}

Now, since a nearly optimal code $\calC_n^*$ (i.e., such that $P_e(W^{(n)},\calC_n^*,q_n)\leq \inf_{\calC_n\in \calG(n,M_n^{\Delta})} P_e(W^{(n)},\calC_n,q_n)+\epsilon$ for arbitrarily small $\epsilon>0$) performs at least as well, in terms of average probability of error, as any code ensemble, i.e., 
\begin{flalign}
&\inf_{\calC_n\in \calG(n,M_n^{\Delta})} P_e(W^{(n)},\calC_n,q_n)\nonumber\\
\leq&\inf_{\mu\in \calP\left((\calX^n)^{M_n^{\Delta}}\right)}\int d\mu(\calC_n) P_e(W^{(n)},\calC_n,q_n)
\end{flalign}
and since the code ensemble which is drawn with independent codewords each drawn according to some $P^{(n)}\in\calP(\calX^n)$ is a special case of a random code, we have
\begin{flalign}
&\inf_{\mu\in \calP((\calX^n)^{M_n^{\Delta}})}\int d\mu(\calC_n) P_e(W^{(n)},\calC_n,q_n)\nonumber\\
\leq &\inf_{P^{(n)}}\mathbb{E}\left[1-\left(1-\Phi_{q_n}(q_n(X^n,Y^n),P^{(n)},Y^n)  \right)^{M_n^{\Delta}-1}\right] 
\end{flalign}
thus, we obtain
\begin{flalign}
&\inf_{\calC_n\in \calG(n,M_n^{\Delta})} P_e(W^{(n)},\calC_n,q_n)\nonumber\\
\leq &\inf_{P^{(n)}}\mathbb{E}\left[1-\left(1-\Phi_{q_n}(q_n(X^n,Y^n),P^{(n)},Y^n)  \right)^{M_n^{\Delta}-1}\right] \nonumber\\
\stackrel{(a)}{\leq} &\inf_{P^{(n)}}\ \mathbb{E} \min\left\{1,M_n^{\Delta}\Phi_{q_n}(q_n(X^n,Y^n),P^{(n)},Y^n)\right\}
\end{flalign}
where $(a)$ follows from the union bound. Now, let $\bP=\{P^{(i)}\}_{i\geq 1}$ be given, and let $\overline{\calA}_n$ denote the set of pairs of $n$-vectors $(\bx,\by)\in\calX^n\times\calY^n$ such that $M_n^{\Delta}\cdot \Phi_{q_n}(q_n(\bx,\by),P^{(n)},\by)\leq 2^{-n\gamma}$
we have 
\begin{flalign}\label{eq: eq skdjbvkdbvkj}
&\mathbb{E} \min\left\{1,M_n^{\Delta}\Phi_{q_n}(q_n(X^n,Y^n),P^{(n)},Y^n)\right\}\nonumber\\
\leq & \mathbb{E} \left( 1\{ (X^n,Y^n)\in \overline{\calA}_n \} \min\left\{1,M_n^{\Delta}\Phi_{q_n}(q_n(X^n,Y^n),P^{(n)},Y^n)\right\} \right)\nonumber\\
& +\mbox{Pr}\left\{\overline{\calA}_n^c\right\}\nonumber\\
\leq  &2^{-n\gamma} +\mbox{Pr}\left\{\overline{\calA}_n^c\right\}\nonumber\\
=& 2^{-n\gamma}+\mbox{Pr}\left\{-\frac{1}{n} \log \left(\Phi_{q_n}(q_n(X^n,Y^n),P^{(n)},Y^n)\right) <\underline{K}_{\bq}(\bW)+\Delta+\delta_n+\gamma\right\},
\end{flalign}
where $\delta_n$ is such that $\frac{1}{n}\log M_n^{\Delta}=\frac{1}{n}\log \left\lceil 2^{n( \underline{K}_{\bq}(\bW)+\Delta)}\right\rceil= 2^{n( \underline{K}_{\bq}(\bW)+\Delta+\delta_n)}$, i.e., $\lim_{n\rightarrow\infty} \delta_n=0$. 
Summarizing (\ref{eq: e-epsilon eq})-(\ref{eq: eq skdjbvkdbvkj}) we obtain
\begin{flalign}
1\leq &\liminf_{n\rightarrow \infty}\inf_{P^{(n)}}\mbox{Pr}\left\{-\frac{1}{n} \log \left(\Phi_{q_n}(q_n(X^n,Y^n),P^{(n)},Y^n)\right) <\underline{K}_{\bq}(\bW)+\Delta+\delta_n+\gamma\right\}.
\end{flalign}
This yields
\begin{flalign}
 \underline{K}_{\bq}(\bW)+\Delta+\gamma\geq \overline{K}_{\bq}(\bW)
\end{flalign}
and since $\Delta+\gamma$ can be made arbitrarily small we have
\begin{flalign}
\underline{K}_{\bq}(\bW)\geq \overline{K}_{\bq}(\bW),
\end{flalign}
and along with the obvious opposite inequality, the equality (\ref{eq: condition STRONG}) follows.
\end{proof}
The next lemma extends Corollary \ref{cr: Upper Bound Expectation} and shows that if the channel satisfies the strong converse property then (\ref{eq: the General formula ksdjhfkjhsk 22}) holds with equality. 
\begin{lemma}\label{lm: Upper Bound Expectation}
If the channel $\bW$ has a finite input alphabet then 
 \begin{flalign}\label{eq: the General formula ksdjhfkjhsk 22 hdfuhuvk}
    \underline{K}_{\bq}(\bW) \leq&
        \liminf_{n\rightarrow\infty}\sup_{P^{(n)}} \frac{1}{n} \mathbb{E}\log\frac{1}{\Phi_{q_n}\left(q_n(X^n,Y^n), P^{(n)},Y^n \right)}\nonumber\\
        \leq&
        \limsup_{n\rightarrow\infty}\sup_{P^{(n)}}  \frac{1}{n} \mathbb{E}\log\frac{1}{\Phi_{q_n}\left(q_n(X^n,Y^n), P^{(n)},Y^n \right)}\leq
            \overline{K}_{\bq}(\bW) ,
\end{flalign}
and if in addition, the channel satisfies the strong converse property, 
 \begin{flalign}\label{eq: the General formula ksdjhfkjhsk 22 hdfuhuvkgfdjgh}
C_{\bq}(\bW) = &     \lim_{n\rightarrow\infty} \sup_{P^{(n)}}\frac{1}{n} \mathbb{E}\log\frac{1}{\Phi_{q_n}\left(q_n(X^n,Y^n), P^{(n)},Y^n \right)},
\end{flalign}
where the supremums in (\ref{eq: the General formula ksdjhfkjhsk 22 hdfuhuvk})-(\ref{eq: the General formula ksdjhfkjhsk 22 hdfuhuvkgfdjgh}) can be restricted to $P^{(n)}$ that is uniform over a subset of $\calX^n$. 
\end{lemma} 
\begin{proof}
The leftmost inequality (\ref{eq: the General formula ksdjhfkjhsk 22 hdfuhuvk}) was established in  (\ref{eq: the inequality expectation without the sup}), and in fact, it holds for more general alphabets. It remains to prove the rightmost inequality assuming the alphabet $\calX^n$ is finite. The proof follows similarly to the proof of \cite[Theorem 3.5.2.]{Han2003}.
Let $\{A_n\}_{n\geq 1}$ be a sequence of non-negative random variables, and denote $\overline{A}\triangleq p\mbox{-}\limsup A_n$. If $ \overline{A}< c<\infty$, then one has for all $\epsilon>0$, 
\begin{flalign}
\mathbb{E}(A_n)=& \mathbb{E}(A_n1\{A_n\geq c \}) + \mathbb{E}(A_n1\{c>A_n\geq \overline{A}+\epsilon \}) + \mathbb{E}(A_n1\{A_n< \overline{A}+\epsilon \}) \nonumber\\
\leq  & \mathbb{E}(A_n1\{A_n\geq c \}) + c\cdot \mathbb{E}(1\{A_n> \overline{A}+\epsilon \}) + ( \overline{A}+\epsilon)\mathbb{E}(1\{A_n< \overline{A}+\epsilon \}) ,
\end{flalign}
by definition of $\overline{A}$, and since $\epsilon$ can be made arbitrarily small, this yields that
\begin{flalign}\label{eq: inequality between liminf and expectationb}
\limsup_{n\rightarrow \infty}\mathbb{E}(A_n)\leq& \overline{A}+ \limsup_{n\rightarrow \infty}\mathbb{E}(A_n1\{A_n\geq c \}) .
\end{flalign}
Now, observe that $-\frac{1}{n}\log\Phi_{q_n}(q_n(X^n,Y^n), P^{(n)},Y^n)\leq -\frac{1}{n}\log P^{(n)}(X^n)\triangleq Z_n$ and take $c=|\calX|+\epsilon$ (which is an upper bound on $p\mbox{-}\limsup -\frac{1}{n}\log P^{(n)}(X^n)$ \cite[Theorem 1.7.2.]{Han2003}). It was proved in \cite[Theorem 3.5.2.]{Han2003} that
\begin{flalign}\label{eq: inequalitkjdfshkvfhdkjhv}
 \limsup_{n\rightarrow \infty}\mathbb{E}(Z_n1\{Z_n\geq |\calX|+\epsilon \})=0 ,
\end{flalign}
and thus, the rightmost inequality of (\ref{eq: the General formula ksdjhfkjhsk 22 hdfuhuvk}) follows. The equality (\ref{eq: the General formula ksdjhfkjhsk 22 hdfuhuvkgfdjgh}) follows from Theorem \ref{th: condition STRONG}. 

\end{proof}

\section{Conclusion}\label{sc: Conclusion}

This paper presents a derivation of a general formula for the mismatch $\bq$-capacity, $C_{\bq}(\bW)$, of the channel $\bW$. 
The general capacity formula is given in terms of the supremum over input distributions sequence of the limit inferior in probability of the exponent of the conditional error probability given the channel input and output in a single drawing of another codeword uniformly over the codebook. 
We provide two proofs for the upper bound on $C_{\bq}(\bW)$:  
The first proof is based on an 
extension of the Verd\'{u}-Han upper bound for the general channel capacity formula. 
The second proof is based on lower bounding (up to a vanishing quantity) the average error probability of a rate-$R$ codebook $\calC$ by the average error probability of the ensemble random code of rate $R-\epsilon$ whose codewords are drawn independently over $\calC$. 

Comparing the general capacity formula applied to the matched metric and the Verd\'{u}-Han channel capacity formula yields an interesting identity between the supremum over input distribution sequence of the limit inferior in probability of the two sequences of random variables (a) $\frac{1}{n}\log \frac{W^{(n)}(Y^n|X^n)}{P_{Y^n}(Y^n)}$ - the mutual information density rate of $(X^n,Y^n)$, the channel input and output, and (b) $-\frac{1}{n}\log (\Phi_{W^{(n)}}(W^{(n)}(Y^n|X^n), P^{(n)},Y^n))$ - the exponent of the conditional error probability given $(X^n,Y^n)$ in a single drawing of another codeword $\tilde{X}^n$ uniformly over the codebook in dependently of $(X^n,Y^n)$. 

Using the insight gained from the derivation of the general capacity formula, we derive 
two max-min upper bounds on the capacity in terms of supremum over input processes of the infimum over a class of channels of the resulting spectral inf-mutual information rates.
A lower bound on the mismatch capacity of the channel $\bW$ with a non-negative decoding metric is derived, which is tight in the matched case. 
We further provide necessary and sufficient conditions for a channel to have a strong converse. 
We study the closely related problem of threshold mismatched decoding, and obtain a general expression for the threshold mismatch capacity and the constant threshold mismatch capacity. The erasures-only general capacity formula is established as a special case.

Another contribution of this paper is a proof of the Csisz\'{a}r and Narayan's conjecture \cite[Open Problem 6]{CsiszarNarayan95}, i.e., that the product space improvement of the random coding lower bound, $C_q^{(\infty)}(W)$, is indeed the mismatch capacity of the DMC $W$ with bounded additive decoding metric. We conclude by proving that in the DMC case, the constant threshold mismatch capacity is equal to the mismatch capacity and by deriving an identity between the two expressions.

\appendix
\subsection{Proof of Theorem \ref{eq: conjecture proof Theorem}}\label{ap: Proof of DMC Theorem}

As mentioned before, the fact that $C_q^{(\infty)}(W)$ is an achievable rate was pointed out in \cite{CsiszarNarayan95} as a generalization of the random coding bound $C_q^{(1)}(W)$ for the product channel $W^K$ with input and output alphabets $\calX^K$ and $\calY^K$, respectively. 
The proof of the upper bound is divided into 9 steps: 

\noindent{\bf Step 1:} First, observe that without loss of asymptotic optimality, one can assume that the codebook contains codewords that  lie in a single type-class. To be more precise, for any given $(n,2^{nR},\epsilon_n)$-code, one can find an $(n,2^{nR'},2\epsilon_n)$ sub-code, where 
\begin{flalign}
R'&= R-\frac{(|\calX|-1)\log(2(n+1))}{n},
\end{flalign}
whose codewords lie in a single type-class. To realize this, first expurgate half of the codewords that have the highest probability of error. This results in a codebook whose maximal error probability is upper bounded by $2\epsilon_n$. Now, pick the dominant type-class (in the sense that its intersection with the remaining codewords is the largest), this leaves at least $1/(n+1)^{(|\calX|-1)}$ of the codewords with maximal (and thus also average) probability of error upper bounded by $2\epsilon_n$. 

\noindent{\bf Step 2:} Next, it is easily verified, using the union bound, that if 
$\calC_n$ is an $(n,2^{nR'},2\epsilon_n)$-code, 
then applying $\calC_n$ repeatedly $K_n=o\left(\epsilon_n^{-1}\right)$ times over the DMC, say 
\begin{flalign}\label{eq: Kn  requirement 1}
K_n=\epsilon_n^{-1/2}
\end{flalign}
times, results in average probability of error not exceeding $K_n\times 2\epsilon_n=2\epsilon_n^{1/2}$. 
For convenience denote
\begin{flalign}
N\triangleq nK_n.
\end{flalign}

Thus, without loss of asymptotic optimality, we can replace the codebook $\calC_n$ with the codebook $\calC_{N}^{prod}$ which is a $K_n$-times concatenation of the code $\calC_n$. In other words, the message $m\in\{1,...,2^{NR'}\}$ is split to $K_n$ sub-messages, $(m_1,...,m_{K_n})$, such that $m_i\in\{1,....,2^{nR'}\}$, $\forall i$, and each message is mapped to $\calX^n$ using the original codebook $\calC_n$. Note that $\calC_{N}^{prod}$ and $\calC_n$ share the same rate.

From Steps 1 and 2 it follows that without loss of asymptotic optimality, we can assume that the codebook is a concatenation of $K_n$ uses of an $(n,2^{nR'},2\epsilon_n)$-code whose codewords lie in a single type-class. 

%STEP 3
\noindent{\bf Step 3:} Note that by construction, also all the codewords of the codebook $\calC_{N}^{prod}$ lie in a single type-class of sequences of length $N$. 

%STEP 4
\noindent {\bf Step 4:} Now, let $X^N$ be the output of the encoder which uses the codebook $\calC_{N}^{prod}$ and let $Y^N$ be the output of the channel $W^N$ when fed by $X^N$.
We next show that the condition (\ref{eq: first condition for corollary2 a}) is always satisfied for a memoryless channel $W^{(N)}=W^N$ with $\tau_n=  \mathbb{E}_{P^{(N)}\times W^{(N)}}(q_N) +\Delta$ for all $\Delta>0$ and vanishingly small $\zeta_N=\epsilon_{2,N}$, i.e., 
\begin{flalign}\label{eq: zeta epsilon2 N}
\mbox{Pr}\left\{ q_N(X^N,Y^N)\geq   \mathbb{E}_{P^{(N)}\times W^{(N)}}(q_N) +\Delta \right\}\leq \epsilon_{2,N},
\end{flalign}
where
\begin{flalign} \label{eq:  epsilon 2N dfn} 
 \epsilon_{2,N}\triangleq &(n+1)^{|\calX||\calY|}\exp\left\{ -n\min_{\hat{P}_{y^N|x^N}: |  \mathbb{E}_{\hat{P}_{x^N,y^N}}(q)-   \mathbb{E}_{\hat{P}_{x^N}\times W}(q) |\geq \Delta  } D\left( \hat{P}_{y^N|x^N}\| W|\hat{P}_{x^N} \right)\right\}.
\end{flalign} 
The following claim states this straightforward argument.  Recall the definition of $\calP_{CC}(\calX,N)$ (\ref{eq: P_CC dfn}).
\begin{claim}\label{eq: DMC convergence claim}
 If $P^{(N)}\in\calP_{CC}(\calX,N)$ then
 \begin{flalign}\label{eq: the equation of Claim 1}
&\mbox{Pr}\left\{ \left|q_N(X^N,Y^N) -  \mathbb{E}(q_N(X^N,Y^N)) \right|\geq \Delta \right\}\leq \epsilon_{2,N},
\end{flalign}
where $ \epsilon_{2,N}$ is defined in (\ref{eq:  epsilon 2N dfn}).
 \end{claim}
 \begin{proof}
To prove the claim, pick an arbitrary $x^N$ that lies in the type-class over which $P^{(N)}$ is defined, note that by symmetry, $  \mathbb{E}_{P^{(N)}\times W^{(N)}}(q_N)=  \mathbb{E}(q_N(X^N,Y^N))=  \mathbb{E}(q_N(X^N,Y^N)|X^N=x^N)=  \mathbb{E}_{\hat{P}_{x^N}\times W}(q)$ and thus, 
\begin{flalign}
&\mbox{Pr}\left\{ \left|q_N(X^N,Y^N) -  \mathbb{E}(q_N(X^N,Y^N)) \right|\geq \Delta \right\}\nonumber\\
=&\mbox{Pr}\left\{ \left|q_N(X^N,Y^N)-  \mathbb{E}_{\hat{P}_{x^N}\times W}(q) \right|\geq \Delta \right\}\nonumber\\
=&\mbox{Pr}\left\{ \left|q_N(x^N,Y^N) -  \mathbb{E}_{\hat{P}_{x^N}\times W}(q) \right|\geq \Delta |X^n=x^N\right\}\nonumber\\
=&\sum_{\hat{P}_{y^N|x^N}: |  \mathbb{E}_{\hat{P}_{x^N,y^N}}(q)-   \mathbb{E}_{\hat{P}_{x^N}\times W}(q) |\geq \Delta  }
\mbox{Pr}
\left\{T(\hat{P}_{y^N|x^N} )|X^n=x^N\right\}\nonumber\\
\leq& \epsilon_{2,N}\rightarrow 0.
\end{flalign} 
\end{proof}
%STEP 5
\noindent {\bf Step 5:} 
From Steps 1-3 we know that $P_e(W^N,\calC_N^{prod},q_N)\leq 2\epsilon_n^{1/2}$. 
Let $\tilde{W}^{(N)}$ be a channel which satisfies
\begin{flalign}
& \mbox{Pr}\left\{q_N(X^N,\tilde{Y}^{(N)}) \leq   \mathbb{E}(q_N(X^N,Y^N) )+\Delta \right\} \leq \epsilon_{1,N} \label{eq: conditions 2}\\
&P_{Y^N}= P_{\tilde{Y}^N},\label{eq: conditions 3}
\end{flalign}
for some $\Delta>0$ and $\epsilon_{1,N} >0$, where $\tilde{Y}^N$ is the output of $\tilde{W}^{(N)}$ when fed by $X^N$, which is uniformly distributed over $\calC_N^{prod}$. 
Note that (\ref{eq: zeta epsilon2 N}), (\ref{eq: conditions 2}), (\ref{eq: conditions 3}) are in fact the conditions (\ref{eq: first condition for corollary2 a})-(\ref{eq: first condition for corollary}) required for Theorem \ref{cr: conditions corollary} to hold, and hence
\begin{flalign}
 P_e(\tilde{W}^{(N)},\calC_N^{prod},q_N)\leq & P_e(W^N,\calC_N^{prod},q_N)+\epsilon_{1,N}+\epsilon_{2,N}\nonumber\\
 \leq& 2\epsilon_n^{1/2}+\epsilon_{1,N}+\epsilon_{2,N}\triangleq \bar{\epsilon}_N .\label{eq: bar epsilon dfn}
\end{flalign}

%STEP 6
\noindent {\bf Step 6:} 
Denote the following set of block memoryless distributions
\begin{flalign}
\calP_N^{prod} \triangleq & \left\{P^{(N)} \in \calP\left(\calX^{N}\right):\; P^{(N)}=\prod_{k=0}^{K_n-1}\bP^{(k)},\;  \forall k \; \bP^{(k)}\in\calP_{CC}(\calX,n)\right\},
\end{flalign}
where $\calP_{CC}(\calX,n)$ is defined in (\ref{eq: P_CC dfn}). 

Now, from the Steps 1-5, we can invoke Fano's Inequality for the channel $\tilde{W}^{(n)}$ and this yields 
\begin{flalign}\label{eq: W_b first app}
&R-\frac{(|\calX|-1)\log(2(n+1))}{n}\nonumber\\
\leq&\max_{P^{(N)}\in \calP_N^{prod}} \min_{\tilde{W}^{(N)}\in \calW_b \left(P^{(N)} ,\epsilon_{N,1}\right) } \frac{1}{N}I(X^N;\tilde{Y}^{N})+\frac{1}{N}+R\cdot \bar{\epsilon}_N,
\end{flalign}
where $\left(X^N,\tilde{Y}^N\right)\sim P^{(N)}\times \tilde{W}^{(N)}$, $\bar{\epsilon}_N$ is defined in (\ref{eq: bar epsilon dfn}) as a function of $(\epsilon_{N,1},\epsilon_{2,N})$ with $\epsilon_{2,N}$ defined in (\ref{eq:  epsilon 2N dfn}), and where
\begin{flalign}\label{eq: dfn calWb}
\calW_b \left(P^{(n)},\epsilon_{N,1}\right)  \triangleq &\left\{\tilde {W}^{(N)}:    \mathbb{E}_{P^{(N)}\times  \tilde{W}^{(N)}}\left(1\left\{q_N(X^N,\tilde{Y}^{(N)}) \leq c_ {N} +\Delta\right\} \right)\leq \epsilon_{N,1}, P_{\tilde{Y}^N}= P_{Y^N}\right\},
\end{flalign}
with
\begin{flalign}\label{eq: cN dfn }
 c_ {N} \triangleq   \mathbb{E}_{P^{(N)}\times  W^{(N)}}\left(q_N(X^N,Y^N)\right).
\end{flalign}
We note that in (\ref{eq: W_b first app}) we can take the minimum over $\tilde{W}^{(N)}\in \calW_b \left(P^{(N)} ,\epsilon_{N,1}\right)$ rather than the infimum because $ \calW_b \left(P^{(N)} ,\epsilon_{N,1}\right)$ is a convex compact set.

%STEP 7
\noindent {\bf Step 7:} 
For $k=0,1,..,K_n-1$ let 
\begin{flalign}
\bX^{(k)}\triangleq (X_{kn+1}, X_{kn+2},....,X_{(k+1)n}),
\end{flalign}
 i.e., 
\begin{flalign}
X^N=(\bX_{(0)},....,\bX_{(K_n-1)}).
\end{flalign}
Denote
\begin{flalign}
c^{(k)}\triangleq   \mathbb{E}_{P^{(N)}\times W^{(N)}  }\left(q_n(\bX^{(k)},\bY^{(k)})\right) ,
\end{flalign}
and note that since the metric is additive we have
 \begin{flalign}
 c_{N}=\frac{1}{K_n}\sum_{k=0}^{K_n-1} c^{(k)},
 \end{flalign}
 where $c_N$ is defined in (\ref{eq: cN dfn }).
 \begin{claim}\label{cl: claim LLN}
 Let $\Delta>0$ and $\delta>\Delta$ be arbitrarily small, and consider a block-wise memoryless channel $\tilde{W}^{(N)}=\prod_{k=0}^{K_n-1} \tilde{\bW}^{(k)}$ (where $\tilde{\bW}^{(k)}$ is a channel from $\calX^n$ to $ \calY^n$) be given. If 
 \begin{flalign}\label{eq: a condition}
    \mathbb{E}_{\bP^{(k)}\times  \tilde{\bW}^{(k)}} q_n\left(\bX^{(k)},\tilde{\bY}^{(k)} \right)\geq c^{(k)}+\delta
  \end{flalign}
   for all $k=1,...,K_n$, then
\begin{flalign}
  \mathbb{E}_{   \prod_{k=0}^{K_n-1} \bP^{(k)}   \times  \tilde{\bW}^{(k)}}\left(1\left\{q_{N}(X^N,\tilde{Y}^{(N)}) \leq c_{N} +\Delta\right\} \right)\leq \epsilon_{1,N}\end{flalign}
 where
 \begin{flalign}\label{eq: epsilon 1N dfn}
\epsilon_{1,N}\triangleq \frac{B^2}{K_n(\delta-\Delta)^2} ,
 \end{flalign}
 and $B$ is the bound on $q$ (\ref{eq: bounded q}). 
 \end{claim}
 \begin{proof}
 The claim follows from the weak law of large numbers (LLN) for a sequence of independent random variables $\{A_k\}_{k=0}^{K_n-1}$ where
\begin{flalign}
A_k= q_n\left(\bX^{(k)},\tilde{\bY}^{(k)} \right)
\end{flalign}
whose expectations are not necessarily identical. 
Recall the definition of $K_n$ (\ref{eq: Kn  requirement 1}). Note that we have from the fact that $q$ is bounded by $B$ (\ref{eq: bounded q}) and from (\ref{eq: a condition})
\begin{flalign}
&c^{(k)}+\delta \leq   \mathbb{E}(A_k) \leq B< \infty\nonumber\\
&Var(A_k)\leq B^2 <\infty\; \forall k,
\end{flalign}
and consequently, since $\delta>\Delta$,  
from Chebishev's Inequality we obtain
\begin{flalign}\label{eq: Kn  requirement 2}
 &\mathbb{E}_{   \prod_{k=0}^{K_n-1} \bP^{(k)}   \times  \tilde{\bW}^{(k)}}\left(1\left\{q_{N}(X^N,\tilde{Y}^{(N)}) \leq c_{N} +\Delta\right\} \right)\nonumber\\
=&\mbox{Pr}\left(\frac{1}{K_n} \sum_{k=0}^{K_n-1} \left[A_k-c^{(k)}\right]\leq \Delta \right)\nonumber\\
=&\mbox{Pr}\left(\frac{1}{K_n} \sum_{k=0}^{K_n-1} \left[A_k-c^{(k)}-\delta\right]\leq \Delta-\delta \right)\nonumber\\
\leq&\mbox{Pr}\left(\frac{1}{K_n} \sum_{k=0}^{K_n-1} \left[A_k-  \mathbb{E}(A_k)\right]\leq -(\delta-\Delta) \right)\nonumber\\
\leq&\frac{\left[  \mathbb{E}\left(\frac{1}{K_n} \sum_{k=0}^{K_n-1} \left[A_k-  \mathbb{E}(A_k)\right]\right)^2\right]}{(\delta-\Delta)^2}\nonumber\\
\leq &\frac{B^2}{K_n(\delta-\Delta)^2}\triangleq \epsilon_{1,N}
\end{flalign}
which vanishes since $K_n\rightarrow\infty$.
\end{proof}
Hence, from Claim \ref{cl: claim LLN} we can deduce that for $\delta>\Delta$, 
\begin{flalign}\label{eq: law of large numbers}
&\calW_a\left(\prod_{k=0}^{K_n-1} \bP^{(k)},\delta\right)\nonumber\\
\triangleq &\left\{\tilde{W}^{(N)}=\prod_{k=0}^{K_n-1}\tilde{\bW}^{(k)}:\; \forall k\;  \mathbb{E}_{\bP^{(k)}\times  \tilde{\bW}^{(k)}} q_n\left(\bX^{(k)},\tilde{\bY}^{(k)} \right)\geq c^{(k)}+\delta, P_{\tilde{\bY}^{(k)}}= P_{\bY^{(k)}}\right\}\nonumber\\
\stackrel{(a)}{\subseteq} &\left\{\tilde{W}^{(N)}=\prod_{k=0}^{K_n-1}\tilde{\bW}^{(k)} : \begin{array}{ll} 
  \mathbb{E}_{   \prod_{k=0}^{K_n-1} \bP^{(k)}   \times  \tilde{\bW}^{(k)}}\left(1\left\{q_{N}(X^N,\tilde{Y}^{(N)}) \leq c_{N} +\Delta\right\} \right)\leq \epsilon_{1,N}
,\\\forall k ,\; P_{\tilde{\bY}^{(k)}}= P_{\bY^{(k)}}\end{array} \right\}\nonumber\\
\stackrel{(b)}{\subseteq}&\calW_b\left(\prod_{k=0}^{K_n-1} \bP^{(k)},\epsilon_{1,N}\right)
\end{flalign}
where $\tilde{\bW}^{(k)}$ is a channel from $\calY^n$ to $\calX^n$, $(a)$ follows from Claim \ref{cl: claim LLN} and $\epsilon_{1,N}$ is defined in (\ref{eq: epsilon 1N dfn}), and $(b)$ follows by definition of $\calW_b\left(\prod_{k=0}^{K_n-1} \bP^{(k)},\epsilon_{1,N}\right)$ (see (\ref{eq: dfn calWb})) and by the fact that the l.h.s\ contains only block-wise memoryless channels.

We get from (\ref{eq: Kn  requirement 2}) and (\ref{eq: Kn  requirement 1})
\begin{flalign}
 \epsilon_{1,N} =\frac{B^2}{K_n(\delta-\Delta)^2}= \frac{B^2\epsilon_n^{1/2}}{(\delta-\Delta)^2}.
\end{flalign}

%STEP 8
\noindent {\bf Step 8:} 
Summarizing Steps 1-7, we obtain that for all $\Delta>0$ and $\delta>\Delta$, an achievable rate $R$ satisfies 
\begin{flalign}
R-\frac{(|\calX|-1)\log(2(n+1))}{n}&\leq \max_{P^{(N)}\in \calP_N^{prod}}  \min_{ \tilde{W}^{(N)}\in\calW_b \left(P^{(N)},\epsilon_{N,1} \right) } \frac{1}{N}I(X^N;\tilde{Y}^N)+\frac{1}{N}+R\cdot \bar{\epsilon}_N\nonumber\\
&\stackrel{(a)}{\leq} \max_{P^{(N)}\in \calP_N^{prod}}  \min_{\tilde{W}^{(N)}\in \calW_a \left(P^{(N)},\delta\right) } \frac{1}{N}I(X^N;\tilde{Y}^N)+\frac{1}{N}+R\cdot \bar{\epsilon}_N \nonumber\\
%&=\max
\end{flalign}
where $\bar{\epsilon}_N$ is defined in (\ref{eq: bar epsilon dfn}), and $(a)$ follows from (\ref{eq: law of large numbers}), i.e., the fact that $ \calW_a \left(\prod_{k=0}^{K_n-1}\bP^{(k)},\delta\right)  \subseteq\calW_b \left(\prod_{k=0}^{K_n-1}\bP^{(k)},\epsilon_{N,1}\right)$.

Now, by definition of $\calW_a \left(P^{(N)},\delta\right) $, $\tilde{W}^{(N)}$ is block-wise memoryless, and therefore
\begin{flalign}
%& 
 \frac{1}{N}I(X^N;\tilde{Y}^N)\leq\sum_{k=0}^{K_n-1} \frac{1}{N}I(\bX^{(k)};\tilde{\bY}^{(k)})
\end{flalign}
Defining $\epsilon_n'\triangleq \frac{1}{N}+R\cdot \bar{\epsilon}_N+\frac{(|\calX|-1)\log(2(n+1))}{n}$ this yields
\begin{flalign}
R&\leq 
\max_{P^{(N)}\in \calP_N^{prod}}  \min_{\tilde{W}^{(N)}\in \calW_a \left(P^{(N)},\delta\right) } \frac{1}{N}I(X^N;\tilde{Y}^N)+\epsilon_n'\nonumber\\
&\leq \max_{P^{(N)}\in \calP_N^{prod}}  \min_{\tilde{W}^{(N)}\in \calW_a \left(P^{(N)},\delta\right) } \sum_{k=0}^{K_n-1} \frac{1}{N}I(\bX^{(k)};\tilde{\bY}^{(k)})+\epsilon_n'\nonumber\\
&= \max_{P^{(n)}\in\calP_{CC}(\calX,n)} \min_{\tilde{W}^{(n)}: \mathbb{E}_{P^{(n)}\times \tilde{W}^{(n)}}(q_n) \geq  \mathbb{E}_{P^{(n)}\times W^{(n)}}(q_n)+\delta, P_{\tilde{Y^n}}= P_{Y^n}} \frac{1}{n}I_{P^{(n)}\times \tilde{W}^{(n)}}(X^n;\tilde{Y}^n)+\epsilon_n',\label{eq: removal of delta}
\end{flalign}
where the last equality follows since the optimizations are in fact performed over block memoryless sources and channels. 

% STEP 9
\noindent {\bf Step 9:} 
Next, we wish to take the limit of the right hand side of (\ref{eq: removal of delta}) as $\delta\rightarrow 0$ to obtain $C_q^{(\infty)}(W)$. 
Before taking the limit, we will treat the case in which the set over which the minimization is performed is empty. 

To this aim, denote
\begin{flalign} 
\mathbb{E}_{P^{(N)}\times W^N}(q_N(X^n,Y^n))=& \mathbb{E}_{\hat{P}_{X^N}}\max_{y\in\calY}q(X,y) -\delta_1\nonumber\\
\triangleq &
q_{\max}(\hat{P}_{X^N}) -\delta_1,
 .\label{eq: continuity assumption for the DMC bbb}
\end{flalign}
and note that since $P^{(N)}\in\calP_{CC}(\calX,N)$, $q_{\max}(\hat{P}_{X^N}) $ is deterministic. 
In this case, since $P^{(N)}\in\calP_{CC}(\calX,N)$, we have
\begin{flalign}\label{eq: zeta epsilon2 N b}
\mbox{Pr}\left\{ q_N(X^N,Y^N)>  q_{\max}(\hat{P}_{X^N}) \right\}=0.
\end{flalign}
Let $\tilde{W}^{(N)}$ be a channel which satisfies 
\begin{flalign}\label{eq: zeta epsilon2 N b c d }
&\mbox{Pr}\left\{ q_N(X^N,\tilde{Y}^N)< q_{\max}(\hat{P}_{X^N})\right\}=0\nonumber\\
& P_{Y^N}=P_{\tilde{Y}^N}.
\end{flalign}
We thus obtain the conditions (\ref{eq: first condition for corollary2 a})-(\ref{eq: first condition for corollary}) required for Theorem \ref{cr: conditions corollary} to hold, 
with $n$ in lieu of $N$, $\tau_N= q_{\max}(\hat{P}_{X^N})$, $\zeta_N=0$, $\eta_N=0$, 
and hence
\begin{flalign}
 P_e(\tilde{W}^{(N)},\calC_N^{prod},q_N)\leq & P_e(W^N,\calC_N^{prod},q_N).
\end{flalign}
Finally, by the additivity of the metric $q_n$, we have that if 
\begin{flalign}\label{eq: condition for prod channel}
\forall k=1,...,K_n, \;\mathbb{E}_{\bP^{(k)}\times \tilde{\bW}^{(k)}}q_n(\bX^{(k)},\tilde{\bY}^{(k)})= \mathbb{E}_{\hat{P}_{\bX^{(k)}}}\max_{y\in\calY}q(X,y)=  q_{\max}(\hat{P}_{\bX^{(k)}})
\end{flalign}
then
\begin{flalign}
\mbox{Pr}\left\{ q_N(X^N,\tilde{Y}^N)<   q_{\max}(\hat{P}_{X^N}) \right\}=0, 
\end{flalign}
i.e., $q_N(X^N,\tilde{Y}^N)=q_{\max}(\hat{P}_{X^N}) $, and hence, we can repeat the derivation similarly to that of Step 8 by replacing $\calW_b(P^{(N)},\epsilon_{N,1})$ with the set of channels $\tilde{W}^{(N)}$ satisfying conditions (\ref{eq: zeta epsilon2 N b c d }), and by replacing $\calW_a(P^{(N)},\delta)$ with the block-wise memoryless channels $\tilde{W}^{(N)}=\prod_{k=1}^{K_n}\tilde{\bW}^{(k)}$ such that (\ref{eq: condition for prod channel}) holds. Summing up, we obtain
\begin{flalign}
R&\leq 
\max_{P^{(n)}\in\calP_{CC}(\calX,n)} \min_{\tilde{W}^{(n)}: \mathbb{E}(q_n(X^n,\tilde{Y}^n)) \geq  \mathbb{E}_{\hat{P}_{X^n}}\max_{y\in\calY}q(X,y) , P_{\tilde{Y^n}}= P_{Y^n}} \frac{1}{n}I_{P^{(n)}\times \tilde{W}^{(n)}}(X^n;\tilde{Y}^n)+\epsilon_n'.
\end{flalign}
Combined with (\ref{eq: removal of delta}) this yields
\begin{flalign}
R&\leq 
\max_{P^{(n)}\in\calP_{CC}(\calX,n)}\nonumber\\
& \min_{\tilde{W}^{(n)}: \mathbb{E}(q_n(X^n,\tilde{Y}^n)) \geq  \min\{ \mathbb{E}(q_n(X^n,Y^n))+\delta,\mathbb{E}_{\hat{P}_{X^n}}\max_{y\in\calY}q(X,y) \} , P_{\tilde{Y^n}}= P_{Y^n}} \frac{1}{n}I_{P^{(n)}\times \tilde{W}^{(n)}}(X^n;\tilde{Y}^n)+\epsilon_n',\label{eq: last equation}
\end{flalign}
which concludes the proof of Step 9. 

Finally,  the above inequality holds for all $\delta>0$. Since the set over which the minimization is performed is convex and non empty, and $I_P(X;Y),E_P(q)$ are continuous in $P$, 
taking the limit of the right hand side of (\ref{eq: last equation}) as $\delta\rightarrow 0$ yields $C_q^{(\infty)}(W)$ and Theorem \ref{eq: conjecture proof Theorem} follows.

%\bibliographystyle{IEEEtran}
%\bibliography{mismatch_bib}

\begin{thebibliography}{10}
\providecommand{\url}[1]{#1}
\csname url@samestyle\endcsname
\providecommand{\newblock}{\relax}
\providecommand{\bibinfo}[2]{#2}
\providecommand{\BIBentrySTDinterwordspacing}{\spaceskip=0pt\relax}
\providecommand{\BIBentryALTinterwordstretchfactor}{4}
\providecommand{\BIBentryALTinterwordspacing}{\spaceskip=\fontdimen2\font plus
\BIBentryALTinterwordstretchfactor\fontdimen3\font minus
  \fontdimen4\font\relax}
\providecommand{\BIBforeignlanguage}[2]{{%
\expandafter\ifx\csname l@#1\endcsname\relax
\typeout{** WARNING: IEEEtran.bst: No hyphenation pattern has been}%
\typeout{** loaded for the language `#1'. Using the pattern for}%
\typeout{** the default language instead.}%
\else
\language=\csname l@#1\endcsname
\fi
#2}}
\providecommand{\BIBdecl}{\relax}
\BIBdecl

\bibitem{CsiszarNarayan95}
I.~Csisz{\'{a}}r and P.~Narayan, ``Channel capacity for a given decoding
  metric,'' \emph{IEEE Trans. Inf. Theory}, vol.~41, no.~1, pp. 35--43, Jan.
  1995.

\bibitem{CsiszarKorner81graph}
I.~Csisz{\'{a}}r and J.~K{\"o}rner, ``Graph decomposition: A new key to coding
  theorems,'' \emph{IEEE Trans. Inf. Theory}, vol.~27, no.~1, pp. 5--12, 1981.

\bibitem{Hui83}
J.~Hui, ``Fundamental issues of multiple accessing,'' \emph{PhD dissertation,
  MIT}, 1983.

\bibitem{Lapidoth96}
A.~Lapidoth, ``Mismatched decoding and the multiple-access channel,''
  \emph{IEEE Trans. Inf. Theory}, vol.~42, no.~5, pp. 1439--1452, Sept. 1996.

\bibitem{SomekhBaruchMismatch1Arxiv2013}
A.~{Somekh-Baruch}, ``On achievable rates for channels with mismatched
  decoding,'' May 2013, arXiv 1305.0547 [cs.IT].

\bibitem{SomekhBaruchISIT_2013}
A.~Somekh-Baruch, ``On coding schemes for channels with mismatched decoding,''
  in \emph{Proc. Int. Symp. Information Theory, ISITÕ}, Istanbul,Turkey, 2013.

\bibitem{ScarlettMartinezGuilleniFabregasISIT_2013}
J.~Scarlett, L.~Peng, N.~Merhav, A.~Martinez, and A.~{Guill\'{e}n i
  F\`{a}bregas}, ``Superposition codes for mismatched decoding,'' in
  \emph{Proc. Int. Symp. Information Theory, ISITÕ}, Istanbul,Turkey, 2013, pp.
  81--85.

\bibitem{Balakirsky_conference_95}
\BIBentryALTinterwordspacing
V.~Balakirsky, ``Coding theorem for discrete memoryless channels with given
  decision rule,'' in \emph{Algebraic Coding}, ser. Lecture Notes in Computer
  Science, G.~Cohen, A.~Lobstein, G.~{Z\'emor}, and S.~Litsyn, Eds.\hskip 1em
  plus 0.5em minus 0.4em\relax Springer Berlin Heidelberg, 1992, vol. 573, pp.
  142--150. [Online]. Available: \url{http://dx.doi.org/10.1007/BFb0034351}
\BIBentrySTDinterwordspacing

\bibitem{ShamaiKaplan1993information}
G.~Kaplan and S.~Shamai, ``Information rates and error exponents of compound
  channels with application to antipodal signaling in a fading environment,''
  \emph{AEU. Archiv f{\"u}r Elektronik und {\"U}bertragungstechnik}, vol.~47,
  no.~4, pp. 228--239, 1993.

\bibitem{MerhavKaplanLapidothShamai94}
N.~Merhav, G.~Kaplan, A.~Lapidoth, and S.~{Shamai (Shitz)}, ``On information
  rates for mismatched decoders,'' \emph{IEEE Trans. Inf. Theory}, vol.~40,
  no.~6, pp. 1953--1967, Nov. 1994.

\bibitem{LiuHughes96}
Y.-S. Liu and B.~Hughes, ``A new universal random coding bound for the
  multiple-access channel,'' \emph{IEEE Trans. Inf. Theory}, vol.~42, no.~2,
  pp. 376--386, Match 1996.

\bibitem{Lapidoth96b}
A.~Lapidoth, ``Nearest neighbor decoding for additive non-gaussian noise
  channels,'' \emph{IEEE Trans. Inf. Theory}, vol.~42, no.~5, pp. 1520--1529,
  Sept. 1996.

\bibitem{GantiLapidothTelatar2000}
A.~Ganti, A.~Lapidoth, and I.~Telatar, ``Mismatched decoding revisited: general
  alphabets, channels with memory, and the wide-band limit,'' \emph{IEEE Trans.
  Inf. Theory}, vol.~46, no.~7, pp. 2315--2328, Nov. 2000.

\bibitem{ShamaiSason2002}
S.~Shamai and I.~Sason, ``Variations on the gallager bounds, connections, and
  applications,'' \emph{Information Theory, IEEE Transactions on}, vol.~48,
  no.~12, pp. 3029--3051, 2002.

\bibitem{ScarlettFabregas2012}
J.~Scarlett and A.~{Guill\'{e}n i F\`{a}bregas}, ``An achievable error exponent
  for the mismatched multiple-access channel,'' in \emph{50th Annual Allerton
  Conference on Communication, Control, and Computing (Allerton)}, October
  2012, pp. 1975--1982.

\bibitem{ScarlettAlfonsoFabregas2013}
J.~Scarlett, A.~Martinez, and A.~{Guill\'{e}n i F\`{a}bregas}, ``Mismatched
  decoding: finite-length bounds, error exponents and approximations,''
  \emph{arXiv:1303.6166 [cs.IT]}, March 2013.

\bibitem{ScarlettMartinezGuilleniFabregas2012AllertonSU}
------, ``Ensemble-tight error exponents for mismatched decoders,'' in
  \emph{50th Annual Allerton Conference on Communication, Control, and
  Computing (Allerton)}, 2012, pp. 1951--1958.

\bibitem{ScarlettPengMerhavMartinezGuilleniFabregasArxiv2013}
------, ``Expurgated random-coding ensembles: exponents, refinements and
  connections,'' July 2013, arXiv 1307.6679 [cs.IT].

\bibitem{Balakirsky95}
V.~B. Balakirsky, ``A converse coding theorem for mismatched decoding at the
  output of binary-input memoryless channels,'' \emph{IEEE Trans. Inf. Theory},
  vol.~41, no.~6, pp. 1889--1902, Nov. 1995.

\bibitem{VerduHan1994}
S.~Verd\'{u} and T.~S. Han, ``A general formula for channel capacity,''
  \emph{IEEE Trans. Inf. Theory}, vol.~40, no.~4, pp. 1147--1157, July 1994.

\bibitem{Han2003}
T.~S. Han, \emph{Information-spectrum methods in {I}nformation {T}heory.}\hskip
  1em plus 0.5em minus 0.4em\relax Berlin: Springer, 2003.

\bibitem{fischer1978some}
T.~R. Fischer, ``Some remarks on the role of inaccuracy in {ShannonÕs} theory
  of information transmission,'' in \emph{Transactions of the Eighth Prague
  Conference}.\hskip 1em plus 0.5em minus 0.4em\relax Springer, 1978, pp.
  211--226.

\bibitem{SomekhBaruchMerhav07}
A.~Somekh-Baruch and N.~Merhav, ``Achievable error exponents for the private
  fingerprinting game,'' \emph{IEEE Trans. Inf. Theory}, vol.~53, no.~5, pp.
  1827--1838, May 2007.

\bibitem{BunteLapidothSamorodintskyArxiv2013}
C.~Bunte, A.~Lapidoth, and A.~Samorodintsky, ``The {zero-undetected-error}
  capacity approaches the {S}perner capacity,'' September 2013, arXiv 1309.4930
  [cs.IT].

\bibitem{Forney1968Erasure}
J.~Forney, G.D., ``Exponential error bounds for erasure, list, and decision
  feedback schemes,'' \emph{Information Theory, IEEE Transactions on}, vol.~14,
  no.~2, pp. 206--220, 1968.

\end{thebibliography}
% Generated by IEEEtran.bst, version: 1.13 (2008/09/30)

\end{document}